\documentclass[10pt,conference]{IEEEtran}
\usepackage{eurosym}
\usepackage{calc}
\usepackage{amssymb}
\usepackage{amsmath}
\usepackage{amsthm}
\usepackage{algorithm}
\usepackage{algorithmicx}
\usepackage{algpseudocode}
\usepackage{graphicx} 
\usepackage{graphics}
\usepackage{hyperref} 
\usepackage{booktabs}

\usepackage{color}

\graphicspath{{../figures/}}

\newtheoremstyle{mystyle} % Name
{}    % Space above
{}    % Space below
{\nopagebreak}         % Body font
{}        % Indent amount
{\scshape \bfseries}% Theorem head font
{:}        % Punctuation after theorem head
{ }        % Space after theorem head 
{}         % Theorem head spec (can be left empty, meaning 'normal')

\theoremstyle{mystyle}
\newtheorem{definition_}{Definition} %[section]
\newtheorem{example}{Example} %[section] 
\newtheorem{theorem}{Theorem} %[section]
\newtheorem{corollary}{Corollary}%[section] 
\newtheorem{myrule}{Rule} %[section]

\newcommand{\rte}{$rt_{\text{est}}$}
\newcommand{\rtb}{$rt_{\text{est}}$ }

\newcommand{\crb}{$cr_{\text{eff}}$ }
\newcommand{\Pe}{O$|$R$|$P$|$E }
\newcommand{\PeFS}{O$|$R$|$P$|$E}

\graphicspath{{../figures/}}

\begin{document}

\title{\textbf{O$|$R$|$P$|$E - A Data Semantics Driven Concurrency Control
Mechanism with Run-time Adaptation}}

\author{
\IEEEauthorblockN{
	Tim Lessner\IEEEauthorrefmark{1},
	Fritz Laux\IEEEauthorrefmark{2},
	Thomas M Connolly\IEEEauthorrefmark{3}
}
\IEEEauthorblockA{
	\IEEEauthorrefmark{1}freiheit.com technologies gmbh, Hamburg, Germany \\ Email:
	tim.lessner@freiheit.com} \IEEEauthorblockA{\IEEEauthorrefmark{2}Reutlingen University,
	Reutlingen, Germany\\
	Email: fritz.laux@reutlingen-university.de}
	\IEEEauthorblockA{\IEEEauthorrefmark{3}University of the West of Scotland,
	Paisley, UK\\
	Email: thomas.connolly@uws.ac.uk}
}

\maketitle

\begin{abstract}
This paper presents a concurrency control mechanism that does not follow a 'one concurrency control mechanism fits all needs' strategy.
With the presented mechanism a transaction runs under several concurrency
control mechanisms and the appropriate one is chosen based on the accessed data.
For this purpose, the data is divided into four classes based on its access type
and usage (semantics). Class $O$ (the optimistic class) implements a
first-committer-wins strategy, class $R$ (the reconciliation class) implements a
first-n-committers-win strategy, class $P$ (the pessimistic class)
implements a first-reader-wins strategy, and class $E$ (the escrow
class) implements a first-n-readers-win strategy. Accordingly, the
model is called \PeFS. 
The selected concurrency control mechanism may be automatically adapted at run-time according to the current load or a known usage profile. This run-time adaptation allows \Pe to balance the commit rate and the response time even under changing conditions.
\Pe outperforms the Snapshot Isolation concurrency control in terms of response time by a factor of approximately 4.5 under heavy transactional load (4000 concurrent transactions). As consequence, the degree of concurrency is 3.2 times higher. 
\end{abstract}


\section{Introduction}
\label{sec:introduction}
The drawbacks of existing concurrency control (CC) mechanisms are that
pessimistic concurrency control (PCC) is likely to block transactions and is
prone to deadlocks, optimistic concurrency control (OCC) may experience a sudden
decrease in the commit rate if contention increases.
Snapshot Isolation (SI) better supports query processing since transactions
generally operate on snapshots and also prevents read anomalies, but depending
on the implementation of SI, either pessimistic or optimistic, it is also
subject to the previously mentioned drawbacks of PCC or OCC. Semantics based CC
(SCC) %such as the mechanism proposed in \cite{Laux2009} 
remedies some problems
of PCC or OCC. It performs well under contention, reduces the blocking time, and
better supports disconnected operations. However, its applicability is limited
since data and transactions have to comply with specific properties such as the
commutativity of operations. In addition to the previously mentioned drawbacks,
neither PCC nor OCC nor SCC support long-lived and disconnected data processing.
However, these properties are essential to achieve scalability in Web-based and loosely coupled applications. 
Another challenge is that in real-life scenarios often the data usage profile changes over time (e.g. stock refill in the morning, selling goods during business hours, housekeeping during closing hours) which calls for a dynamic CC-mechanism.

This paper extends a mechanism presented in \cite{LessnerDBKDA2015} and originally introduced in \cite{Lessner2014} that
combines OCC, PCC, and SCC and steps away from the \textbf{`one concurrency
control mechanism fits all needs'} strategy. Instead, the CC mechanism is chosen
depending on the data usage. While the original \Pe model assigns the appropriate CC-mechanism statically, this paper addresses a dynamic adaptation of the CC-mechanism due to sudden changes of the system load. To address scalability, the mechanism was designed with a focus on long-lived and disconnected data processing.

Consider, for example, the wholesale scenario as presented in the TPC-C \cite{TPCC.2010}.
With PCC using shared and exclusive locks, the likelihood of deadlocks increases
for hot spot fields such as the stock's quantity or the account's debit or
credit. If transactions are long-lived, PCC is even worse since deadlocks
manifest during write time and a significant amount of work is likely to be lost
\cite{Thomasian1998a} \cite{Lessner2014}. With OCC, deadlocks cannot occur. However,
hot-spot fields like an account's debit or credit would experience many version
validation failures under high load causing the restart of a transaction. Like
PCC, validation failures manifest during the write-phase of a transaction and a
significant amount of work is likely to be lost.
Both PCC and OCC cannot ensure that modifications attempted during a
transaction's read-phase will prevail during the write-phase. Whereas PCC is
prone to deadlocks, OCC is prone to its optimistic
nature itself.

\Pe resolves these drawbacks and data can be classified in CC classes. For
example, customer data such as the address or password can be controlled by a
PCC that uses exclusive locks only \cite{JimGray.1993}. Such a rigorous measure
ensures ownership of data and should be used if data is modified that belongs to
one transaction. For example, account data or master data should not be modified
concurrently and given the importance of this data a rigorous isolation is
justified. The debit or credit of an account can be classified in CC class $R$,
which guarantees no lost updates and no constraint violations. Such a guarantee
is often sufficient for hot-spot fields. Class $E$ can be used to access an
item's stock, for example.
Class $E$ is able to handle use cases such as reservations.
It should be used if during the read-phase a guarantee is required that the
changes will succeed during the write-phase.
Class $O$ is the default class. It avoids blocking and under normal load it
represents a good trade-off between commit and abort-rate.

Section \ref{sec:model} defines these four CC classes with different data access
strategies used by our mechanism. In the case of a conflict, class $O$ implements a
\textbf{first-committer-wins} strategy, class $R$ implements a
\textbf{first-n-committers-win} strategy, class $P$ implements a
\textbf{first-reader-wins} strategy, and class $E$ implements a
\textbf{first-n-readers-win} strategy. The number $n$ is determined by the
semantics of the accessed data, e.g., by database constraints. According to the
classes, the mechanism is called \PeFS. The ``$|$'' indicates the demarcation
between data.

Section \ref{sec:correctness} proofs the correctness of the model. Section
\ref{sec:prototypical_ref_impl} briefly describes the prototype implementation.
Section \ref{sec:performance_study} highlights some advantages of \PeFS, because
it provides an application flexibility in choosing the best suitable CC
mechanism and thereby significantly increases the commit rate and outperforms
optimistic SI. 
The run-time adaptation mechanism and its adaptation rules are presented in Section \ref{sec:perform_adapt}. In the following Section \ref{sec:performance_study} a prototype implementation is tested with various workloads. The results are discussed and the behavior is illustrated with time diagrams. Section \ref{sec:related_work} summarizes related work and compares it to our model.
Finally, the paper draws some conclusions and provides an outlook (see Section \ref{sec:outlook}) to future work.

\section{Model}
\label{sec:model}
The model relies on disconnected transactions and 4 CC classes, which are defined in the following.
\subsection{Transaction}
To support long-lived and disconnected data processing, which both supports
scalability, \Pe models a transaction as a disconnected transaction $\tau$, with
separate read- and write-phase, i.e., no further read after the first write
operation (see Definition \ref{def:disconnected_transaction}, taken from
\cite{Lessner2014}). To disallow blind writes, \Pe guarantees that in addition to
the value of a field, the version of a data field has to be read, too.

\begin{definition_} Disconnected Transaction:
\label{def:disconnected_transaction}
	\begin{enumerate}
		\item Let $ta$ be a flat transaction that is defined as a pair $ta=(OP,<)$ where
		$OP$ is a finite set of steps of the form $r(x)$ or $w(x)$ and ${<} (\subseteq OP
		\times OP)$ is a partial order.
		
		\item A disconnected transaction $\tau=(TA^R, TA^W)$ consists of two
		disjoint sets of transactions. $TA^R=\{ta^R_1,\ldots,ta^R_i\}$ to read and
		$TA^W=\{ta^W_1,\ldots,ta^W_j\}$ to write the proposed modifications back.
		
		\item A transaction has to read any data item $x$ before being allowed to
		modify $x$ (no blind writes).
		
		\item If a transaction only reads data it has to be labeled as read only.
	\end{enumerate}
\end{definition_}

\subsection{CC Classes}
Class $O$ is the default class and is implemented  by an optimistic SI mechanism,
which is advantageous since reads do not block writes and non-repeatable or
phantom phenomena do not happen. However, SI is not fully serializable
\cite{Fekete2005a} \cite{Cahill2009}.

As stated, the drawback of optimistic mechanisms prevails if load increases,
because many transactions may abort during their validation at commit time. An
abort at commit time is expensive, because significant amount of work might be
lost. A circumstance particularly crucial for long-lived transactions (see
\cite{Lessner2014}).

Regarding the strategy, optimistic SI follows a ``first-committer-wins''
semantics revealing another drawback of $O$. It is the lack of an option
allowing a transaction to explicitly run as an owner of some data.
Consider, for example, the private data of a user such as its password or
address. A validation failure should be prevented by all means, since it would
mean that at least two transactions try to concurrently update private data.
Although technically this is a reasonable state, for this kind of data a
pessimistic approach that acquires all locks at read time is more appropriate.
Such a mechanism follows a ``first-reader-wins'' (ownership) semantics and
directly leads to class $P$. The acquisition of exclusive locks at read time
prevents deadlocks during write time. To prevent deadlocks at all, a strict
sequential access and preclaiming (all locks appear before the first read) or
sorted read-sets are possible mechanisms. Which mechanism is chosen to prevent
or resolve deadlocks is unimportant regarding the correctness of \Pe (see
Section \ref{sec:correctness}). Preclaiming has its drawbacks concerning the
time a lock has to be acquired. Sorted read-sets may be unfeasible due to
limitations of the storage layer or chosen index structure. The prototype (see
Section \ref{sec:prototypical_ref_impl}) uses a Wait-For-Graph to prevent
deadlocks during the read-phase of a transaction. Also, during our experiments
(see Section \ref{sec:performance_study}) the number of deadlocks was
considerably small, because data classified in $P$ should have no concurrent
modifications by definition.

The decision if a data item is classified as $O$ or $P$ is based on the following
properties \cite{Lessner2014}:

\begin{enumerate}
\item \textit{Mostly read} ($mr$): Is the data item mostly read?
If 'Yes', there is no need for restrictive measures and the data item should by
classified for optimistic validation. A low conflict probability is assumed.
\item \textit{Frequently written} ($fw$): $fw$ is the opposite of $mr$.
\item \textit{unknown} ($un$): It means neither $mr$ nor $fw$ apply, i.e., it is
unknown whether an item is mostly read or written or approximately even.
\item \textit{Ownership} ($ow$): if accessing a data item should explicitly
cause the transaction to own this item for its lifetime?
\end{enumerate}

\begin{example} Classify data items in class $O$ and $P$
(taken from \cite{Lessner2014}).
\label{example:classifiy_data_items_in_O_and_P}

This example is based on the TPC-C \cite{TPCC.2010} benchmark and its
``New-Order'' transaction. Note that an additional table Account has been
introduced to keep track about a customer's bookings (column debit and credit).
It also defines an overdraft limit (column limit). The following tables are used
in our example: Customer (id, name, surname), Stock (StockId, ItemId, quantity),
Account (AcctNo, debit, credit, limit), and Item (ItemId, name, unit, price).
Table \ref{table:example_customer_order_first_task} shows an initial
classification. 

Attributes \textit{name}, \textit{surname}, and \textit{id} of a customer are
expected to be mostly read, but if modified by a transaction it should
definitively be the owner. The \textit{id} of a customer, like all ids, is
expected to become modified rarely. If the \textit{id} becomes modified,
ownership is required. In principal, all business keys should be classified in
$P$, because they are owned by the application provider (see Rule
\ref{rule:premises_CC_class}, 1)).

\textit{Stock.quantity} is expected to become modified frequently ($fw$) and to
prevent the situation where an item was marked as available during the read
phase, but at commit time the item is no longer available due to concurrent
transactions, it is also marked as $ow$. For the time being, however,
\textit{quantity} will be classified as an ambiguity (see also Rule
\ref{rule:premises_CC_class}, 3)), which will be discussed below.

The \textit{Account.credit} and \textit{Account.debit} of a customer's account
might be accessed frequently depending on a customer's activity and $un$ is a good
choice. However, since multiple transactions might concurrently update the
balance, and an owner is hardly identifiable, $\neg ow$ is chosen. So, it is
also an ambiguity (see Rule~ \ref{rule:premises_CC_class}, 3)).

The \textit{Account.limit} is the overdraft limit of a customer and expected to
be mostly read, hence, $mr$ is a good choice. Since it is neither owned by the
customer nor by others, $\neg ow$ is a good choice (see Rule~
\ref{rule:premises_CC_class}, 2)).

Assuming the application is a high frequency trading application,
\textit{Item.Price} might quickly become a bottleneck. An exact prediction is
not possible though, hence, $un$ is a good choice. Property $ow$ would not be a
good choice, because transactions of different components ($dc$) might
simultaneously calculate the price (see Rule~ \ref{rule:premises_CC_class},
3)).
\end{example}

\begin{table}[h] \centering
\caption{Classification of Example
\ref{example:classifiy_data_items_in_O_and_P}}
    \begin{tabular}{| c | c | c | c | c | c |}
    \cline{1-6}
    $x$ & $mr$ & $fw$ & $un$ & $ow$ & CC class \\ \hline
	Customer.name			&		1&0&0&1&$P$		
	\\ \hline
	Customer.surname				&		1&0&0&1&$P$		
	\\ \hline
	Customer.id				&		1&0&0&1&$P$		
	\\ \hline
	Stock.StockId				&		1&0&0&1&$P$	
	\\ \hline
	Stock.ItemId				&		1&0&0&1&$P$	
	\\ \hline
	Stock.quantity		&		0&1&0&1&$\textbf{A}$		
	\\ \hline
	Account.debit				&		0&0&1&0&$\textbf{A}$		
	\\ \hline
	Account.credit				&		0&0&1&0&$\textbf{A}$		
	\\ \hline
	Account.limit				&		0&0&1&0&$\textbf{A}$		
	\\ \hline
	Item.name				&		1&0&0&1&$P$
	\\ \hline
	Item.unit				&		1&0&0&1&$P$
	\\ \hline
	Item.price			&		0&0&1&0&$\textbf{A}$
	\\ \hline
	
    \end{tabular}

\label{table:example_customer_order_first_task}
\end{table}

The ambiguities $A$ of Example \ref{example:classifiy_data_items_in_O_and_P},
see class $A$ in Table \ref{table:example_customer_order_first_task}, highlight
that classes $O$ and $P$ and their properties are not sufficient.
Particularly, hot spot items such as \textit{Stock.quantity} would benefit from
a CC mechanism that allows many winners and resolves the drawbacks of OCC and
PCC.

Laux and Lessner \cite{Laux2009} propose the usage of a mechanism that
reconciles conflicts --class $R$--. Their approach is an optimistic variant of
O'Neil's \cite{O'Neil1986} Transactional Escrow Method (TEM). Both
approaches exploit the commutativity of write operations. If operations commute,
it is irrelevant which operation is applied first as long as the final state can
be calculated (see \cite{Laux2009} \cite{Lessner2014} for further details) and no
constraint is violated.

Unlike TEM, the reconciliation mechanism requires a dependency function.
Consider, for example, two transactions that update an account and both read an
initial amount of 10\EUR, one credits in 20\EUR and the other debits 10\EUR.
Once both have committed, it is relevant that no constraint was violated at any
time and the final amount has to be 20\EUR.
Usually, a database would write the new state for each transaction causing a
lost update. A dependency function would actually add or subtract the amount
(the delta!) and would always take the latest state as input. In other words,
reconciliation replays the operation in case of a conflict.
However, this is only possible if no further user input is required. In the
example above this means the user wants to credit 20\EUR (or debit 10 \EUR)
independent of the account's amount as long as no constraint is violated!
Another requirement is that each dependency function has to be compensatable
(see also \cite{Lessner2014}).

The reconciliation mechanism \cite{Laux2009} follows a
``first-n-committers-win'' semantics and the number of winners $n$ is solely
determined by constraints. The correctness of the mechanism is proven in
\cite{Laux2009}, which also introduces ``Escrow Serializability'', a notion for
semantic correctness.

TEM grants guarantees to transactions during their read-phase. For example, a
reservation system is able to grant guarantees to a transaction about the
desired number of tickets as long as tickets are available. The consequence is
that transactions need to know their desired update in advance (see
\cite{O'Neil1986} for further details).

Whereas TEM \cite{O'Neil1986} is pessimistic (constraint validation during the
read phase) and works for numerical data only, Reconciliation \cite{Laux2009} is
optimistic (constraint validation during the write phase) and works for any data
as long as a dependency function is known. The proof that $E$, like $R$, is
escrow serializable can be found in \cite{Lessner2014}.

The decision if an item is member of $R$ or $E$ is based on the following
properties:
\begin{enumerate}
\item $con$: Does a constraint exist for this data item?
\item $num$: Is the type of the data item numeric?
\item $com$: Are operations on this data item commutative?
\item $dep$: Is a dependency function known for an operation modifying the data item?
\item $in$: Is user input independence given for an operation modifying the data item?
\item $gua$: Is a guarantee needed that a proposed modification will succeed?
\end{enumerate}

\begin{myrule}Derivation of CC classes for data item $x$
\label{rule:premises_CC_class}
\begin{enumerate}
\item $ow \rightarrow$ classify $x$ in $P$ (identify $P$).
\item $\neg ow \wedge mr \rightarrow$ classify $x$ in $O$ (identify $O$).
\item all other combinations of $ow$ and $mr$: classify $x$ in $A$ (ambiguity).
\item $com \rightarrow$ classify $x$ in $E \vee R$
\begin{enumerate}
\item $(con \wedge num \wedge com \wedge gua) \rightarrow$ classify $x$ in $E$
(identify $E$).
\item $(in \wedge dep \wedge com) \rightarrow$ classify $x \in R$ (identify
$R$).
\end{enumerate}
\item $x \in A \rightarrow$ item $x$ will be eventually
in $O$.
\end{enumerate}
\end{myrule}

\begin{example}[Classification of data items in $R$ and $E$]
\label{example:classifiy_data_items_in_R_and_E}

The ambiguities of Table \ref{table:example_customer_order_first_task} are the
input for this example. Table \ref{table:example_customer_order_second_task}
shows the result of the classification of these ambiguities.

\noindent \verb+Stock.quantity+ has a constraint $value>0$ and is numeric. The
dependency function $dep$ is known too. As stated above, a dependency function
performs a context dependent write. For example, dependency function $d$ would
be $d(x, xread, xnew)=x+(xnew-xread)$. User input independence $in$ is not
given. If placing the order fails at the end, a replay would also fail. So,
class $R$ is not an option. Since an order requires a guarantee that the
requested amount of items remains available, Rule~ \ref{rule:premises_CC_class}, 4a) applies.

\verb+Account.credit+ and \verb+Account.debit+ are classified as $R$.
Property $dep$ is known, because operations are either additions or
subtractions. Property $in$ is given, because the account has to be updated if
the order is placed and no constraint is violated. As the updates follow a
dependency function they can be reconciled and should not raise an exception.
Again, only a constraint violation such as an overdraft can cause the abort.
Rule~ \ref{rule:premises_CC_class}, 4b) applies.

\verb+Item.price+ depends on a variety of parameters including the last price
itself. As a result, a price update might not be commutative. \verb+Item.price+
remains ambiguous and remains in $O$, because $O$ is the default class. Rule~
\ref{rule:premises_CC_class}, 5) applies.
 
\end{example}

\begin{table}
\centering
\caption{Illustrative classification of ambiguities of Example
\ref{example:classifiy_data_items_in_O_and_P}.}
\label{table:example_customer_order_second_task}
\centering
\resizebox{\columnwidth}{!}{%
    \begin{tabular}{| c | c | c | c | c | c | c | c |}
	    \cline{1-8}
	    $x$ & $con$ & $com$ & $num$ & $dep$ & $in$ & $gua$ & CC class 
	    \\ \hline
		Stock.quantity				&		1&1&1&1&0&1&$E$	
		\\ \hline
		Account.credit				&		1&1&1&1&1&0&$R$	
		\\ \hline
		Account.debit				&		1&1&1&1&1&0&$R$	
		\\ \hline
		Item.price				&		0&0&1&1&0&0&$O$
		\\ \hline
	\end{tabular}%
}
\end{table}

\section{Correctness}
\label{sec:correctness}
A transaction potentially runs under four different CC mechanisms. Due to the CC
classes' individual semantics, each class has a different notion for a conflict,
too. In any case, two read operations are never in conflict because read operations do not alter the database state and hence are commutative \cite{Kifer2005}.

Usually, a conflict is given if two operations access the same data item and the
corresponding transaction overlap in their execution time, and at least one
operation writes the data item \cite{JimGray.1993}. Whereas for $O$ and $P$ this is a
correct definition of a conflict, for $R$ and $E$ it is not, because both can
resolve certain write conflicts. The resolution of conflicts is a key aspect and
advantage of SCC, and SCC questions the seriousness of a conflict. In other
words, the meaning of a read-write or write-write conflict is interpreted. For
$R$ and $E$ only a constraint violation is a conflict. Moreover, the state read
by an operation is assumed to be irrelevant, otherwise commutativity is not
given. It follows that any final serialization graph $SG-R$ and $SG-E$ for class
$R$ and $E$ is non-cyclic because potential conflicts are reconciled (see
\cite{Lessner2014} for a thorough discussion).

For $P$, the common definition of a conflict is correct. If a transaction wants to
modify item $p$ (let $p \in P$), it has to acquire a lock on $p$ during its
read-phase to become the exclusive owner. If not, the transaction does a blind
write, which is disallowed according to Definition
\ref{def:disconnected_transaction}. Hence, every write in $P$ cannot encounter a
concurrent write or read, because if a transaction writes $p$ it has to be the
exclusive owner of $P$.

Consider the following (incorrect) schedule, for example ($disc_i$ and $disc_j$ denote
the disconnect phase of transaction $i$ (resp. $j$) and let $o \in O$ and $p \in P$):
\begin{multline}
r_i(o), r_j(p), r_j(o), disc_j, w_j(o), c_j, r_i(p), disc_i, \\ w_i(p), c_i
\label{eq:cycleSchedule}
\end{multline}

In this schedule transaction $i$ reads $o$ before $j$ modifies $o$ and transaction $j$ reads $p$ ($r_j(p)$) before $i$ writes $p$ ($w_i(p)$). Usually, the ordering of transaction operations are visualized by a precedence graph as in Figure \ref{fig:cycleSG}. 

\begin{definition_}[Serialization Graph (SG)]
Let $S$ be a schedule of transactions. The Serialization Graph (aka Conflict Graph) is a precedence graph where each node represents a transaction and each directed edge between two transactions represents a precedence of conflicting operations \cite{GerhardWeikum.2002} \cite{Garcia-Molina2009} on a data item.
\end{definition_}

It is well known that a transaction schedule is conflict serializable if and only if the SG is acyclic \cite{GerhardWeikum.2002} \cite{AbrahamSilberschatz2011}. If the SG of a transaction schedule includes a cycle then no equivalent serial schedule exists and, therefore, this schedule is not serializable \cite{GerhardWeikum.2002}.  

The above Schedule \ref{eq:cycleSchedule} leads to the following cyclic SG of Figure \ref{fig:cycleSG}.

\begin{figure}[htbp] \centering
\includegraphics[width=1.5in]{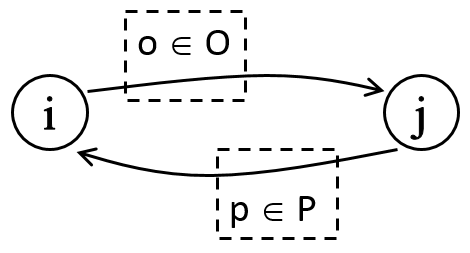}
\caption{The cyclic serialization graph from Schedule (\ref{eq:cycleSchedule}).}
\label{fig:cycleSG}
\end{figure}

Transaction $i$ precedes $j$ in class $O$ and $j$ precedes $i$ in $P$. Having
opposite orders, i.e., $i \rightarrow j$ in one, but $j \rightarrow i$ in
another class violates serializability, because globally $i$ precedes $j$, which
in turn precedes $i$.

A transaction that reads a data item in $O$ has to validate the value at write-time, even if the write is only for an item $p \in P$.  The operation $w_i(p)$ causes a
validation failure on item $o$ because transaction $i$ has read a value of $o$ that transaction $j$ has meanwhile updated. This is a conflict between transactions $i$
and $j$ in $O$ and produces a validation failure. Commit $c_i$ is wrong in the schedule above and would never
happen in \PeFS. Hence, the above schedule looks as follows in \PeFS:
\begin{multline}
r_i(o), r_j(p), r_j(o), disc_j, w_j(o), c_j, r_i(p), disc_i, \\ w_i(p), a_i
\end{multline}
Even a deadlock in $P$ cannot create a cyclic graph between $O$ and $P$, because
at least a write is required to create a conflict in $P$. However, since all
deadlocks can only happen during the read phase of a transaction, no conflict cycle
involving a deadlock can happen in $P$.

Based on these initial findings it is possible to state Theorem
\ref{theorem:globalConf}. The corresponding proof  exploits
that for $R$, $P$, and $E$ the corresponding serialization graphs are
non-cyclic.

\begin{theorem}
\label{theorem:globalConf}
Let $SG-G$ be the global serialization graph, which is the union of $SG-O$,
$SG-R$, $SG-P$, and $SG-E$. The global serialization graph $SG-G$ is non-cyclic
if $SG-O$ is non-cyclic.
\end{theorem}

\begin{proof}[Proof by contradiction]
\label{proof:globalSer}

Given that $ta_i$ is serialized before $ta_j$ $(i \rightarrow j)$ in $SG-O$.
In $P$, no other transaction can access an item in $P$ if transaction $ta_i$ has
read this item. This is the consequence of x-locks during the read-phase used in
class P.
The same argument applies to $ta_j$ as well and it is impossible to have a
serialization order $j \rightarrow i$ in $P$. Since $i$ and $j$ can be
arbitrarily changed there is a contradiction if $i \rightarrow j$ exists in one,
and $j \rightarrow i$ in another class. $SG-R$ and $SG-E$ are negligible because
any conflict is finally reconciled and both serialization graphs are non-cyclic.
\end{proof}

\begin{corollary}
\label{cor:globalSer}
$SG-O$ sets the global serialization order for $P$.
\end{corollary}

If a $ta$ does not modify data in $O$, then $P$ sets the order. If a $ta$ does
not modify data in $P$, then $R$ sets the order, because it is prone to
validation conflicts as opposed to $E$ that already has a guarantee to succeed.

\section{Prototype Reference Implementation of O\textbar R\textbar P\textbar E}
\label{sec:prototypical_ref_impl}

The prototype of \Pe is not a full database system. From a fully operational database the backup and recovery functions are missing. Both functions do not functionally influence the CC mechanism. There is only a negative effect on the performance during backup or recovery. This applies in a similar way for any database management system with a single CC mechanism.   

It was implemented using the
JAVA programming language and Figure \ref{fig:orpe_arch_prototype} illustrates
its architecture. A client API provides access to the data and depending on the
operation's type, read or write, the operation is executed by a dedicated pool.
Pools ``Reads'' and ``Writes'' represent an read- and write-lane. 
In addition, a pool to handle the termination (commit and abort) has been implemented. 
Pools' reads and writes handle all incoming and outgoing operations and the classification has been placed directly into the index. 
Depending on an item's classification the corresponding CC mechanism is plugged in. 
This placement allows to decide about the CC mechanism with a single read operation, which imposes an negligible overhead. 
Once an item has been read or written, the additional pools' ``read-callback'' and
``write-callback'' deliver the results back to the clients. 
A Pool WFG (Wait-for-Graph) is used to handle access to the WFG. %hä??
Deadlocks may occur during the read-phase of a transaction if the transaction accesses data items in class $P$. 
Deadlocks can only occur in class $P$ during the read-phase, because lock acquisition is not globally ordered.

Having separate pools and callbacks to handle incoming and outgoing operations
means that the prototype supports disconnected transactions, because the entire
communication is asynchronous. Figure \ref{fig:orpe_arch_msg_flow} illustrates
the message flow within the prototype. A read operation is passed to the
``Reads'' pool. Each read is executed asynchronously and the complete read set
is sent back to the client via a dedicated callback pool. To support
asynchronous writes, a write operation is passed to the ``Writes'' pool and if
all writes have been applied the write set is sent back to the client. Clients
always sent their complete write-set.

Data is kept solely in memory and no data is written to disk unless
the operating system needs to swap data to disk due to memory limitations. The only
output to disk is to write logging events that are used for performance
evaluation. Other functionality that has been implemented includes:

\begin{itemize}
\item CC mechanisms $O$, $R$, $P$ and $E$,
\item The prototype supports constraints,
\item The prototype supports item selects, range-selects, updates, and inserts. The
deletion of an item is implemented as update that invalidates a data item.
\item A WFG implementation.
\end{itemize}

\begin{figure}[htbp] \centering
\includegraphics[width=3in]{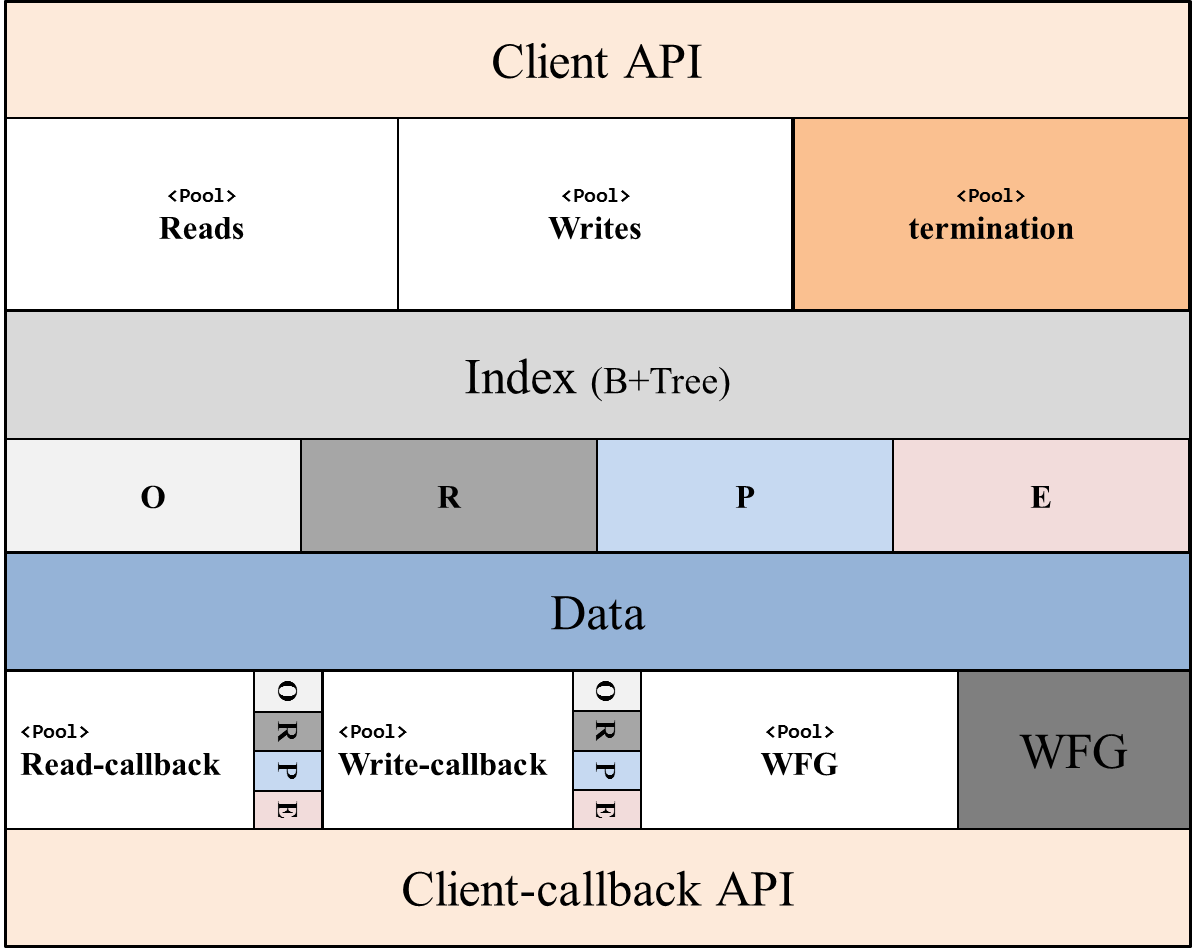}
\caption{Architecture of the prototype.}
\label{fig:orpe_arch_prototype}
\end{figure}

\begin{figure}[htbp] \centering
\includegraphics[width=3in]{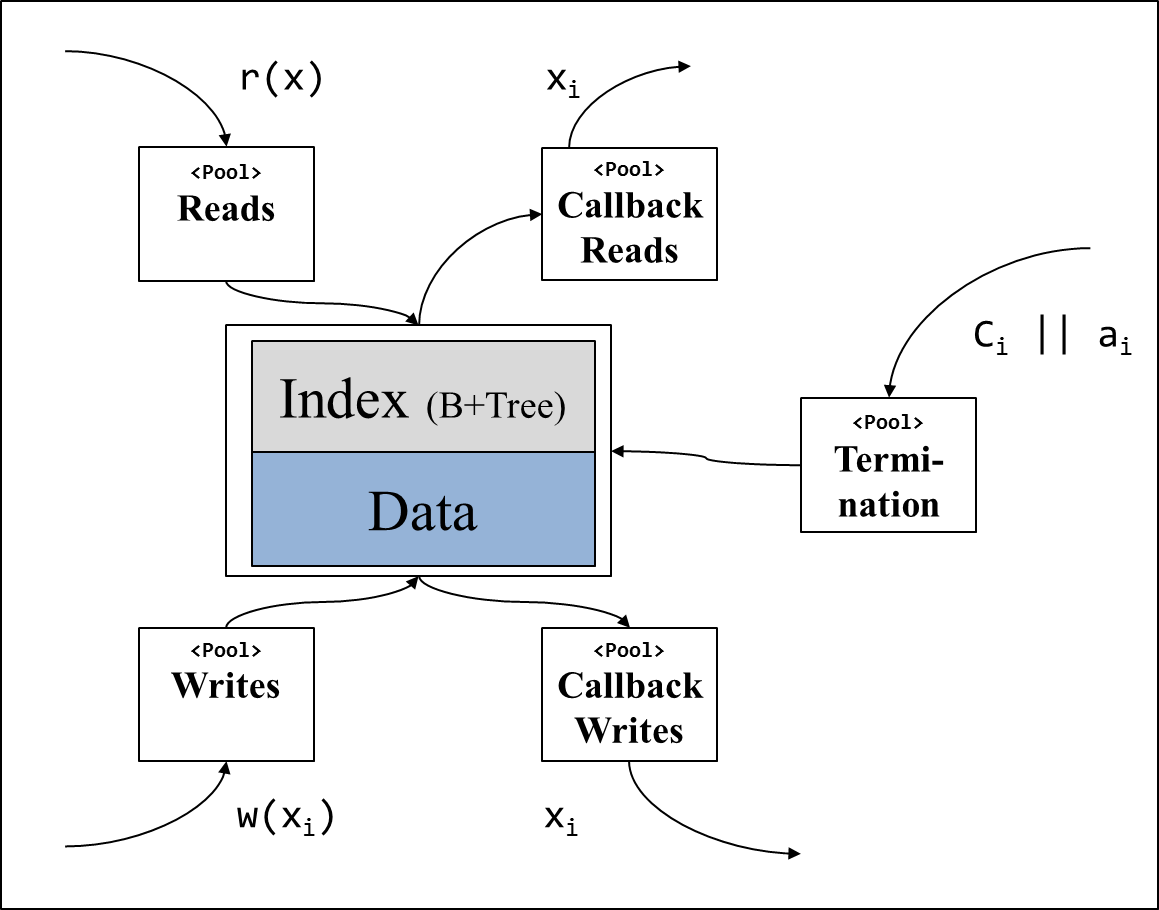}
\caption{Message flow of the prototype.}
\label{fig:orpe_arch_msg_flow}
\end{figure}

\section{Performance Study with Static Data Class Assignment}
\label{sec:performance_study}
The performance study has been carried out based on the prototype presented in
the previous section (Section \ref{sec:prototypical_ref_impl}). As benchmark,
the TPC-C++ benchmark \cite{Cahill2009} has been chosen, because we also conducted a
study comparing \Pe with Serializable SI, which is beyond the scope of this
paper.

The data used for this study is similar to those of Examples \ref{example:classifiy_data_items_in_O_and_P} and
\ref{example:classifiy_data_items_in_R_and_E}. Each data item was statically  assigned to a CC-Class as shown in Table \ref{tab:tpc_orpe_classification}. Aspects of a dynamic assignment and its performance effects will be studied in the next section.

The performance study measures the response-time (resp. - time), the abort rate
(ab-rate), the commits per second, and the degree of concurrency (deg. conc.). The
degree of concurrency is the quotient of the serial estimated execution time over the elapsed time of the experiment.
In addition, the arrival rate $\lambda$ of new transactions has been varied to
be set to the optimum (minimized abort rate and response time, maximized degree of
concurrency). This optimum $\lambda$ has been taken to conduct fair
and calibrated comparisons. Each experiment has been repeated three times and
the mean value is reported. Values refer to the execution of a transaction mix
--deck-- (42 New Order-, 42 Payment-, 4 Delivery-, 4 Credit check-, 4 Update
Stock Level-, and 4 Read Stock Level - transactions see \cite{Cahill2009} \cite{TPCC.2010} \cite{Lessner2014}).

\begin{table}%[]
  \centering
  \caption{TPC-C: classification of data items.}
    \begin{tabular}{rrrr}
    \toprule
    \textbf{Item} & \textbf{CC Class} & \textbf{operation} \\
    \midrule
    Customer & $P$     & read \\
    CustomerCredit & $P$     & update \\
    CustomerBalance & $R$     & read \\
    Customer & $P$     & read \\
    CustomerBalance & $R$     & update \\
    Customer & $P$     & read \\
    CustomerCredit & $P$     & read \\
    StockQuantity & $E$     & update \\
	Customer & $P$     & read \\
    CustomerBalance & $R$     & update \\
    WarehouseYTD & $R$     & update \\
    DistrictYTD & $R$     & update \\
	StockQuantity & $E$     & read only \\
	StockQuantity & $E$     & update \\
    \bottomrule
    \end{tabular}%
  \label{tab:tpc_orpe_classification}%
\end{table}%

Figure \ref{fig:tpc_si} illustrates the abort rate and degree of concurrency
for SI under full contention and shows the drawbacks of
optimistic SI: the higher the number of concurrent transactions, the higher the abort rate.
Also, the system starts thrashing if the degree of concurrency drops below one,
which is the point where a serial execution outperforms a concurrent. Table
\ref{tab:tpc_orpe_measured_values} shows that for SI and \Pe with the same $\lambda$ (tests \#1-6 and \#10-15) the response-time increases
with larger $\lambda$, which is expected and normal behavior. 
The direct comparison reveals that \Pe has a $3 - 38$ times better response time, which shows that SI is over-strained for a workload of $\lambda \geq 200$. 
For $\lambda = 1000$ tas/sec the response time is about 3 times higher for SI and the degree of concurrency is only half compared to \PeFS.  
A good degree of concurrency with
a low abort rate is given by $\lambda = 133$ (see Table
\ref{tab:tpc_orpe_measured_values} \#3).

\begin{figure}[htbp]
\centering%
\includegraphics[width=3.4 in]{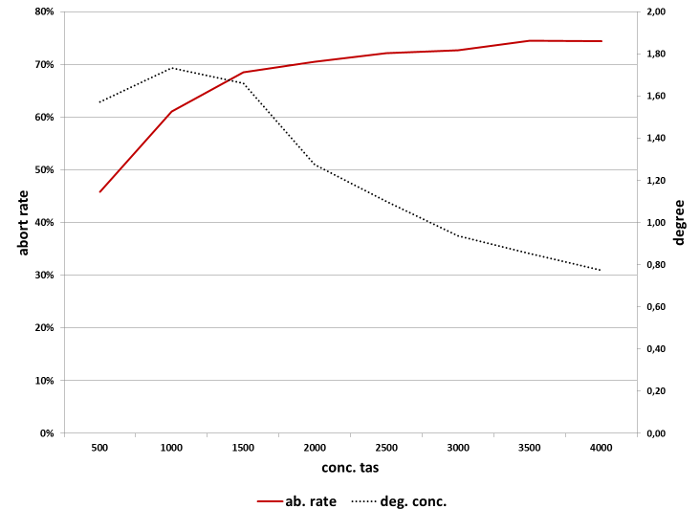}
\caption{TPC-C++, optimistic SI (class $O$), abort rate and degree of
concurrency.}
\label{fig:tpc_si}
\end{figure}

Figure \ref{fig:tpc_orpe} shows the response-time and degree of concurrency for
\Pe for increasing $\lambda$. Unlike SI, \Pe has no aborts caused by
serialization or validation conflicts due to the classification of hot-spot data
items in $R$ or $E$, which prevents $ww$-conflicts. As shown by Figure
\ref{fig:tpc_orpe}, \Pe has its best degree with $\lambda=1000$ transactions per
second achieving 227 commits per second (see Table
\ref{tab:tpc_orpe_measured_values}, \#15). 

\begin{figure}[htbp]
\centering%
\includegraphics[width=3.4 in]{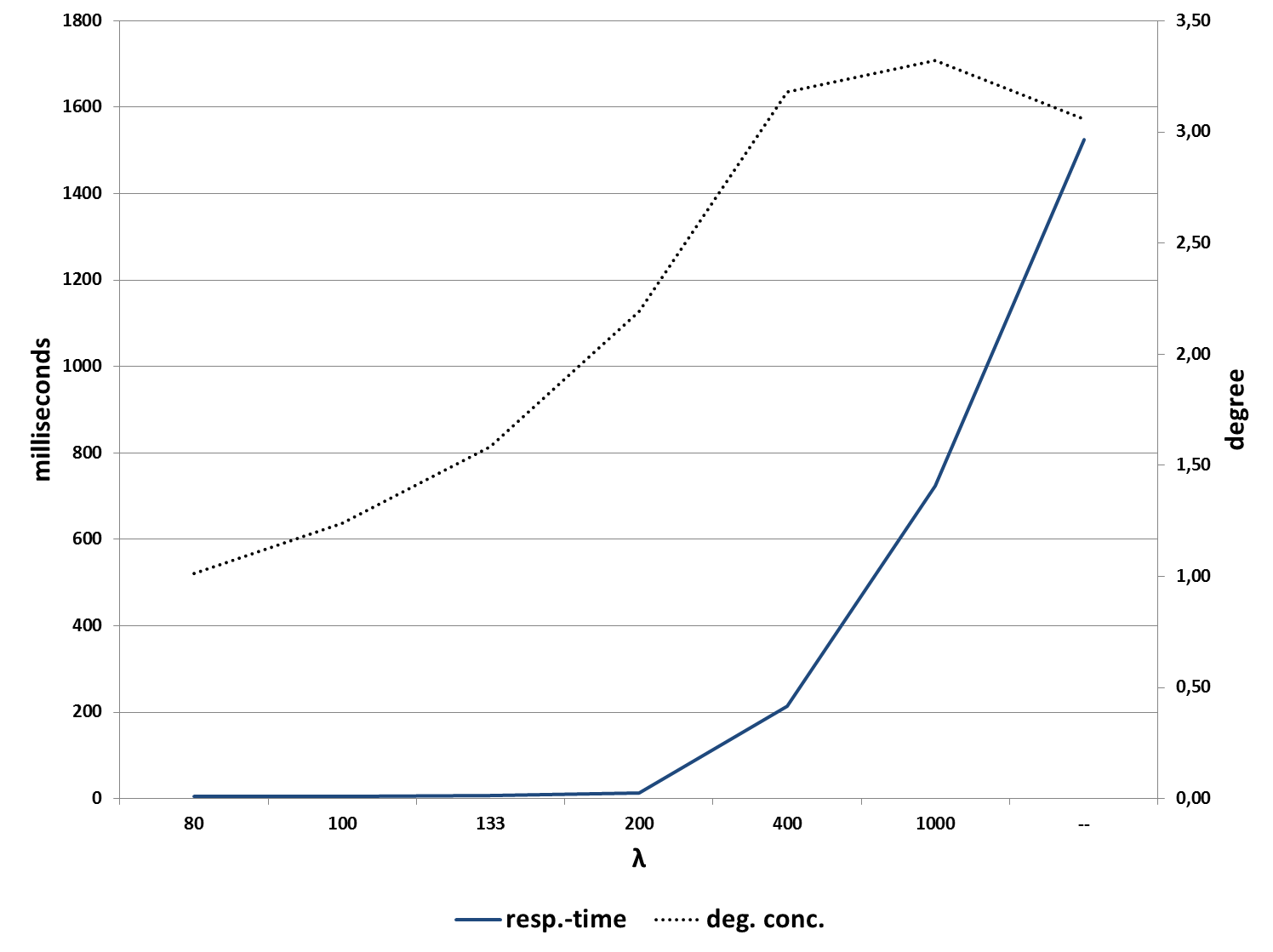}
\caption{Response time and degree of concurrency for increasing $\lambda$ for
\Pe.}
\label{fig:tpc_orpe}
\end{figure}

\begin{figure}[ht]
\includegraphics[width=3.4 in]{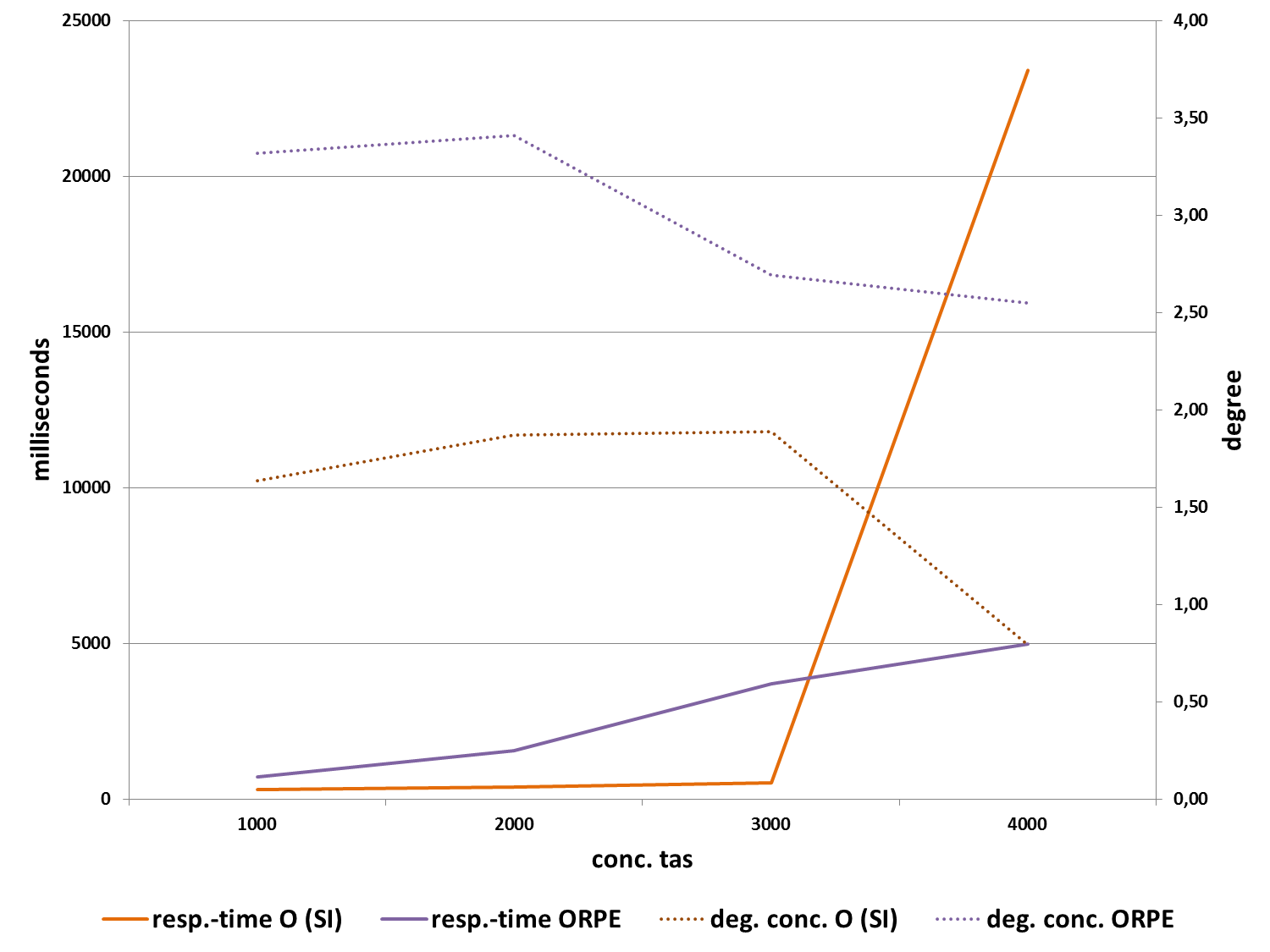}
\caption{TPC-C++, SI and \Pe: response-time and degree of concurrency for
$\lambda=133$ (SI) and $\lambda=1000$ (\Pe).}
\label{fig:tpc_o_and_orpe}
\end{figure}

The comparison of \Pe and SI uses $\lambda=133$ (Table
\ref{tab:tpc_orpe_measured_values} \#3, and \#7-9) for SI and $\lambda=1000$
(Table \ref{tab:tpc_orpe_measured_values} \#15-18) for \Pe. For SI, $\lambda=133$
was considered as being the best trade-off with respect to the degree of
concurrency, $\lambda=1000$ was considered as being the best trade-off for
\PeFS. 

Figure \ref{fig:tpc_o_and_orpe} illustrates the degree and the
response-time for data of class $O$ with SI and \Pe if both use the $\lambda$ which reflect the best trade-off.
As the figure shows, SI has a better response-time for 1000, 2000, and 3000
concurrent transactions, but then suddenly undergoes thrashing and the
response-time grows exponentially. However, \Pe shows a moderate and stable
increase of the response-time even for 4000 concurrent transactions. 

With a workload of $2000$ transactions the degree of concurrency is $3.41$ for \Pe versus $1.87$ for SI. The average response time is only $388$ msec for SI and $1551$ msec for \PeFS.
It would be wrong to conclude that SI has a better performance than \Pe because for a comparison $\lambda$ has to be taken into account. In the test \Pe had a 7.5 times higher transaction arrival rate than SI ($\lambda=1000$ as opposed to $\lambda=133$ for SI). 
At $4000$ concurrent transactions \Pe outperforms SI in terms of response time by a factor of $3.7$ (see Figure \ref{fig:tpc_o_and_orpe}) and the degree of concurrency is $2.6$ times better.
Hence, under high contention \Pe has the lowest abort rate and considering the trade-off between concurrency and response time, \Pe outperforms SI significantly. Furthermore, its abort rate is nearly independent of the contention.

\begin{table}
  \caption{Measured values of experiments \#1-18.}
  \centering
    \begin{tabular}{rrrrrrrrrrrrr}
    \toprule
    \# & tas & $\lambda$ & resp.-time & ab. rate & commits & deg.\\
    &  &  & & & /second & conc. \\
    \midrule
    SI\\
    \midrule
    1   & 1000  & 80    & 43    & 2\%   & 71    & 1,39 \\
    2   & 1000  & 100   & 84    & 3\%   & 80    & 1,57 \\
    3   & 1000  & 133   & 309   & 5\%   & 82    & 1,63 \\
    4   & 1000  & 200   & 1640  & 20\%  & 62    & 1,50 \\
    5   & 1000  & 400   & 2091  & 26\%  & 61    & 1,57 \\
    6   & 1000  & 1000  & 2464  & 27\%  & 62    & 1,61 \\
    7   & 2000  & 133   & 388   & 9\%   & 90    & 1,87 \\
    8   & 3000  & 133   & 522   & 8\%   & 91    & 1,89 \\
    9   & 4000  & 133   & 23416 & 46\%  & 22    & 0,79 \\
    \midrule
    \Pe \\
    \midrule
    10   & 1000  & 80    & 5     & 4\%   & 69   & 1,01 \\
    11   & 1000  & 100   & 5     & 4\%   & 85   & 1,24 \\
    12   & 1000  & 133   & 8     & 4\%   & 108   & 1,58 \\
    13   & 1000  & 200   & 14    & 4\%   & 150   & 2,19 \\
    14   & 1000  & 400   & 213   & 4\%   & 217   & 3,18 \\
    15   & 1000  & 1000  & 724   & 4\%   & 227   & 3,32 \\
    16   & 2000  & 1000  & 1551  & 4\%   & 234   & 3,41 \\
    17   & 3000  & 1000  & 3704  & 4\%   & 184   & 2,69 \\
    18   & 4000  & 1000  & 4968  & 5\%   & 174   & 2,55 \\
    \bottomrule
    \end{tabular}%
  \label{tab:tpc_orpe_measured_values}%
\end{table}%

\section{Run-time Adaption}
\label{sec:runtime_adaption}
The attempt to manually classify data may finally result in ambiguous
classification where default class $O$ applies (see Rule~
\ref{rule:premises_CC_class}, 5)). But, high contention can quickly cause
performance issues for data classified in $O$. Even if class $P$ is more 
expensive, because $P$ requires locking during the read-phase it will lead to a
better performance in this situation as the locking will queue the transactions
and process them successfully.

An automatic and dynamic adaptation of the classification when transactional
load or data usage changes would make the initial classification less critical
and \Pe could choose the optimal CC-mechanism based on the current situation.

A solution for automatic run-time adaptation is presented in this section. It
re-classifies a data items of default class $O$ to class $P$ if the commit rate
drops below an adjustable threshold. With this measure the commit rate increases
again for the price of a longer response time.
When the transactional load decreases and after the commit rate exceeds the
threshold again it switches back to its original class $O$.

Data originally classified in $P$ will not be re-classified to $O$ when the load
is low. This is not feasible, because an item initially in $P$ has to remain in
$P$ due to the item's ownership semantics. An adaption at run-time that results
in $O$ would contradict the ownership semantics since a transaction would no
longer request locks during its read-phase. This is, however, mandatory to
comply with the ownership semantics (see Rule~\ref{rule:premises_CC_class}, 1)).

At a first glance, an adaption between $E \rightarrow R$ seems reasonable if the
probability of an invariant violation (\textit{PIV}) is low. It would save
additional overhead, because invariant conditions in $R$ have not to be validated at read-time, but in $E$.
However, this is only a good decision if contention is low.
To take this decision at high workload will result in a much longer response
time because the response time for class $R$ grows much faster than for class
$E$.
With high contention, the probability of constraint violations increases, but
the exact determination is application dependent.
Classifying a data item in $E$ is only justified if an aborted transaction is
more costly than to retry the transaction, i.e., the transaction needs a
guarantee to succeed which leads to class $E$ from the beginning (Rule~
\ref{rule:premises_CC_class}, 4a)).

\subsection{Adaptation Criteria}
\label{ssec:adapt_criteria}

The run-time adaptation is based on the commit rate $cr$. To measure and analyze
$cr$ a statistical model for the transactional system is necessary.
According to \cite{Kraska2009} \cite{GomezFerro2012} \cite{Osman2012}, a transactional system
is modeled as an open system whose transactional arrival rate is a Poisson
process.
The time between arrivals of transactions is assumed to be independent in
Poisson, which has the advantage that the conflict rate (the term conflict is
stated more precisely below) can be modeled around a single variable $\lambda$
that represents the number of arrivals in relation to the time window. A Poisson
process has a conflict probability density function $PC_x(X=k)$ given by
Equation (\ref{eq:poisson}):

\begin{equation}
PC_x(X=k)=\frac{\lambda^k}{k!}e^{-\lambda}
\label{eq:poisson}
\end{equation}

For example, if on average $100$ transactions arrive within one
Time Window (TW), the probability that $k=50$ transactions access item $x$ within a TW is given by Formula (\ref{eq:poisson}). 
The arrival rate $\lambda$ is
in relation to time, for example, within one second; i.e., for a transaction that accesses $x$ during that second, it means that 
the probability is $PC_x(X >= 2) = \sum_{k=2}^\infty \lambda^k/k!\,e^{-\lambda}$ to encounter other conflicting concurrent transactions.

\begin{figure}
\centering
\includegraphics[scale=0.35]{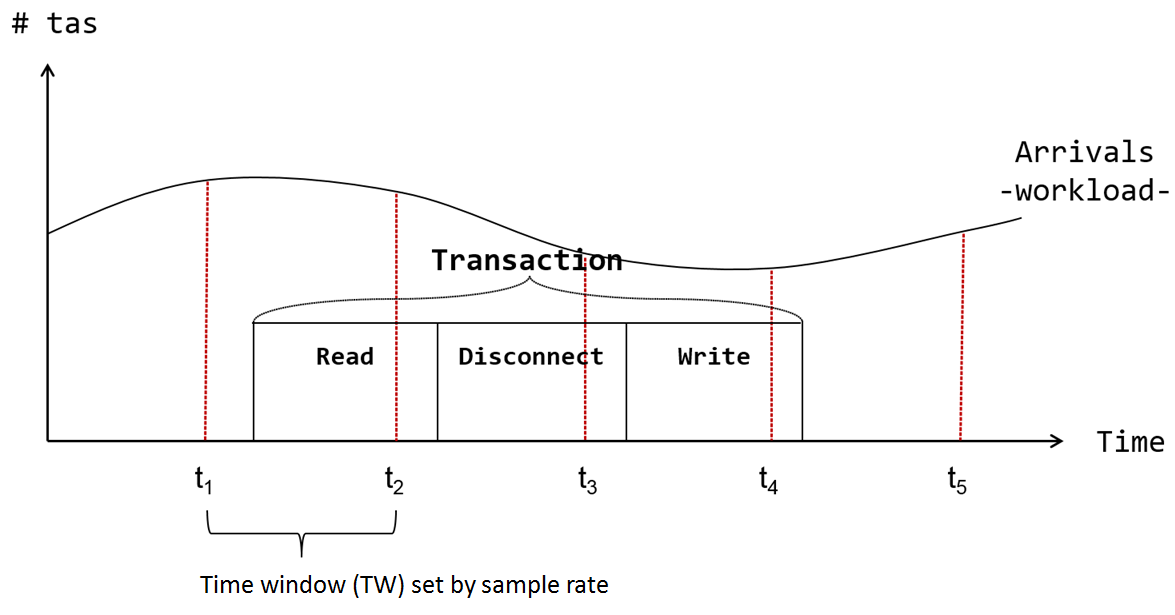}
\caption{Arrivals (workload) and time windows.}
\label{fig:lambda}
\end{figure}

Figure \ref{fig:lambda} illustrates the usage of TW as well as the arrivals
--workload-- in relation to time. The workload is, however, not constant over
the lifetime of a transaction. A constant workload ignores that the workload,
and hence, $\lambda$ might suddenly change in particular if transactions are
long running. Measuring the number of transactions terminating or committing
during a time window are means to detect and react to sudden changes in the
workload, which is an idea borrowed from \cite{Kraska2009}. The length of the TW
defines the sample rate and its sensitivity.

The commit rate $cr$ is used as indicator for the performance of the optimistic
CC-mechanism of class $O$. If the $cr$ drops below a threshold, there are more
aborts due to validation failures and that class $P$ would be a better choice
to increase $cr$.

\begin{equation}
cr := \frac{\mbox{\#committed tas}/\mbox{TW}}{(\mbox{\#terminated tas}-\mbox{\#re-class. aborts})/\mbox{TW}}
\label{eq:commit_rate}
\end{equation}

For each TW the commit rate $cr$ is calculated as fraction of the committed
transaction divided by all terminated transaction without those that were
aborted due to a re-classification.
The commit rate $cr$ is identical to the effective commit rate $cr_{\text{eff}}$
(see Definition \ref{eq:commit_rate_eff}) if no adaptation occurs.
Formula (\ref{eq:commit_rate}) is apparently insensitive to the length of the TW.
But, a longer TW tends to compute smoother $cr$ and it saves measuring overhead.
We used a TW of $100$ msec which delivered a good trade off for the prototype
implementation.

The adaptation policy is given by Rule~ \ref{rule:adapt_O_P}, which uses a
threshold $\gamma$ for the target commit rate and an hysteresis $\delta$ to
avoid constant switching (thrashing) between both classes.
When a data item is re-assigned during an active transaction, the transaction is
aborted when the change is from $O$ to $P$.
In the opposite case, the transaction can continue without conflicts, because
the write-phase will succeed since the data item is already exclusively locked
for that transaction.

\begin{myrule}General Adaptation $O \rightarrow P$
\label{rule:adapt_O_P}
\item Let $cr$ be the commit rate, $\delta$ the hysteresis, and $\gamma$ the
target commit rate. Adaptation is according to the following rules:

\begin{enumerate}
\item When $cr$ decreases and $O$ is the current class for an item $x$: If $cr <
\gamma - \delta$ then $P$ is the new classification of $x$ 

\item When $cr$
increases and $P$ is the current class for an item $x$: If $cr > \gamma + \delta$
then $O$ is the new classification of $x$. 

\item Reclassification during a transaction: \\ a) If a $ta$ reads at the time
when $O$ is the current class, but will write at a time when $P$ is the
current class, $ta$ is aborted (non-avoidable crash) to maintain consistency. \\
b) If a $ta$ reads at a time when the data item is in $P$ and writes when it is
in $O$, the success of the write is guaranteed because the data is exclusively
locked since read-time.
\end{enumerate}
\end{myrule}

Adaptation solely relies on the commit rate $cr$. The arrival rate $\lambda$ and hence
the conflict probability are not measured which would be much more difficult. 
This leverages the decision to use a Poisson distribution for the transaction arrivals.

\begin{figure*}
\centering
\includegraphics[scale=0.6]{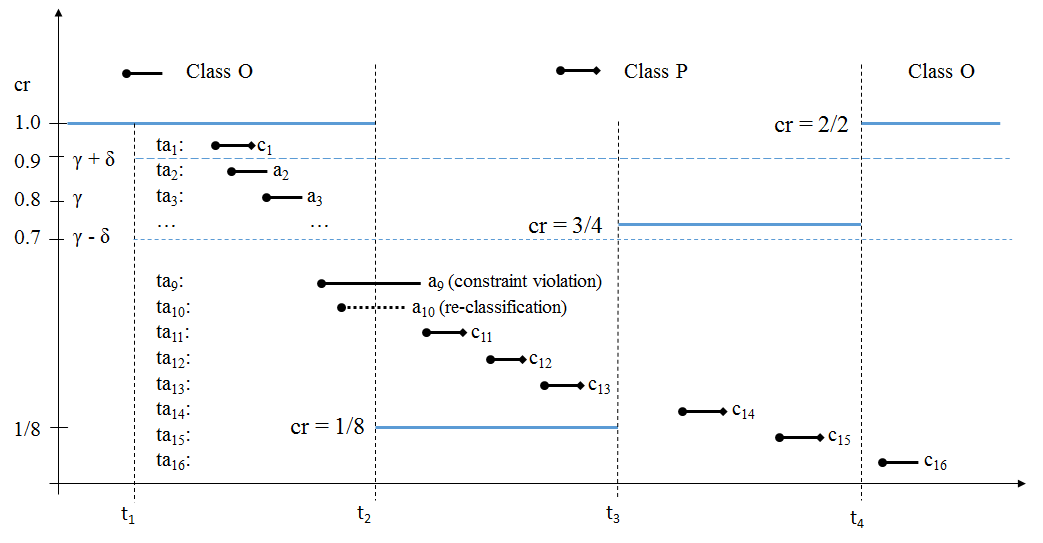}
\caption{Example run-time adaptation scenario with decreasing $cr$ in TW $(t_1,t_2)$, reclassification at $t_2$ to $P$ and increasing $cr$ in TW $(t_2,t_3)$ and $(t_3,t_4)$ and switch back to class $O$ at $t_4$.}
\label{fig:adapt_scenario}
\end{figure*}

Figure \ref{fig:adapt_scenario} illustrates how the adaptation works if
the commit rate decreases and later increases again. During the first TW ($t_2 -
t_1$) the commit rate $cr$ drops to $1/8$ because only one out of 8 transactions
was successful.
Two transactions ($ta_9, ta_{10}$) have not terminated yet.

At the end of epoch 1 the commit rate is compared to $\gamma - \delta$ and as
$cr$ is below the threshold data $x$ is re-classified to $P$.
The transaction $ta_9$ will later abort due to a constraint violation and
$ta_{10}$ has to abort because of the re-classification to $P$.
Now, for the following transactions the locking mechanism for $P$ applies. One
consequence is that $ta_{11}, ta_{12}$, and $ta_{13}$ execute mostly
sequentially.
The commit rate grows in the following TW to $3/4$, but, this is not sufficient
to switch $x$ back to class $O$.
During the third TW ($t_4 - t_3$) the commit rate rises to $cr = 2/2 > \gamma +
\delta$ and the (initial) optimistic CC (class $O$) is re-established.

The following history describes the example of Figure \ref{fig:adapt_scenario} more formally: 
\begin{align*}
H = &
\underbrace{(r_1(x),r_2(x),r_3(x), \ldots, r_{10}(x),
w_1(x), c_1}_{\text{commit rate decreases}}\\
& \underbrace{w_2(x), a_2,w_3(x),a_3,\ldots}_{\text{commit rate decreases}}, \mbox{adapt to P}, a_{10}, a_9, 
\\ &
 \underbrace{l_{11}(x), r_{11}(x), w_{11}(x), c_{11}, l_{12}(x), r_{12}(x),}_{\text{commit rate increases}}
\\ &
\underbrace{ w_{12}(x), c_{12}, l_{13}(x), r_{13}(x), w_{13}(x), c_{13} \ldots }_{\text{commit rate increases}}
\end{align*}

The history $H$ shows in the first phase 10 transactions $ta_1, ta_2, \dots, ta_{10}$ accessing $x$. They first read $x$ ($r_1(x), r_2(x),r_3(x), \ldots, r_{10}(x)$) and then try to write $x$ ($w_1(x), w_2(x), \ldots$). 
In the given scenario only $ta_1$ can commit ($c_1$), all others have to abort ($a_2, a_3, \ldots$) because too many transactions try to concurrently update $x$. 
This leads to a sudden decrease in the commit rate $cr = 1/8$ because only $ta_1$ was successful and $ta_9$ and $ta_{10}$ have not yet updated $x$, i.e., it is still pending. 
If we assume a threshold $\gamma$ of $0.8$ and an hysteresis $\delta$ of $0.1$, then $cr < \gamma - \delta$ which triggers the adaption according to Rule~ \ref{rule:adapt_O_P}, 1).

After adaptation has been
carried out, $ta_{10}$ has to abort (Rule~ \ref{rule:adapt_O_P}, 3a)) if it tries to update $x$. 
The abort $a_{10}$ appears in the history after the adaptation even though the item $x$ is now classified in $P$. Transaction $ta_{10}$ has to abort, because it has not locked $x$ before reading $x$ ($r_{10}(x)$). If $ta_{10}$ would not abort it would  risk a lost-update, because $ta_{10}$ would overwrite the last committed state since $P$ does no version validation. Even with version validation, $ta_{10}$ is very likely to abort, because the probability for a validation failure is high in this situation.

Let assume that transaction $ta_9$ accesses other data beside $x$ and validation fails due to a constraint violation. This leads to an abort of $ta_9$. The distinction of the abort reason is important here as it will be counted for the commit rate.

After the adaptation to $P$ newly arriving transactions apply a locking scheme
for data $x$ which is indicated by $l_{11}, l_{12}, \ldots$. The commit rate
increases again because transactions $ta_{11}, ta_{12}, ta_{13}$ succeed and
commit $c_{11}, c_{12}, c_{13}$. In fact, all following transaction succeed
except those which violate a constraint.

If we choose the Time Window TW to start just before $ta_{11}$ arrives the commit rate $cr$ rises with each committed transaction. Class $O$ is not reestablished at the end of this TW despite that the next $3$ transactions succeed because $cr = 3/4 \leq \gamma + \delta = 0.9$. The class assignment remains unchanged and the following TW ($t_4 - t_3$) will reestablish class $O$ because $cr = 2/2$.

The adaptation mechanism proposed in Rule~ \ref{rule:adapt_O_P} maximizes the commit rate as seen in the previous example. But due to the restrictive locking policy the response time increases as the execution tends to be serial. In the worst case, enduring contention, the growth is exponential. But, what if the maximum
response-time is limited, for example, by Service Level Agreements (SLA) and penalties
apply for exceeding the maximum acceptable response-time? 
The SLA penalties may outweigh the costs for aborts.

In this case maximizing the commit rate as only criteria is not a good strategy since it increases costs. To prevent unacceptable response times a barrier (denoted as $\beta$) is used that regulates the adaptation; i.e., once $\beta$ is reached re-classification to $O$ takes place despite a low commit rate and the abort rate starts to increase which in turn leads to shorter response times for the remaining successful transactions. The concrete value of $\beta$ is
application dependent. Its general purpose is to minimize costs, i.e., if the
abort costs are lower than the costs for exceeding the response-time, more
aborts are acceptable until the ratio turns over.

Application specific requirements that set $\beta$ are out of the paper's scope, but to allow applications to limit the adaptation, $\beta$ is incorporated in \Pe (see Rule~ \ref{rule:adapt_O_P_with_barrier}). 
Applications can now set $\beta$ to limit the response time and, at run-time, continuously
monitor and adapt the achieved commit rate as well as the response time as measured by the applications themselves. Further, applications can increase $\beta$ at run-time appropriately. This way, applications can determine their own equilibrium between commit rate and response-time.

The challenge is the estimation of the expected mean response-time $rt_{\text{est}}$, which
implies to predict the workload. As stated in the previous section, this is
complicated if not impossible in a general and dynamic way. \Pe circumvents this
problem and measures the time between a read and the corresponding write if the current classification is $P$. 
Furthermore, adaptation does no longer calculate $cr$ at the end of the current TW, instead each termination (commit and abort) triggers the adaptation. 
A useful fixed TW is difficult to choose. If TW is too short, the overhead is considerable and degrades performance. If the TW is too long the adaptation is too slow.

To estimate the future workload the terminating transaction snapshots
the lock queue's size if $P$ is the current class. The current queue
size together with the average time between read and write give a good indication for the expected workload. 
Because the transaction has to notify all waiting
transactions about the ongoing unlock and already is the current owner of the lock-queue, there is no need for further synchronization and the overhead is considerably low, but of course exists. It is a price that has to be paid to get run-time adaptation.

Finally, the number of notified transactions multiplied by the average time distance between a read and write is used as an approximation for $rt_{\text{est}}$. 
The rationale is that if $q$ transactions are waiting to execute and the mean time between read and write is \o$(mt)$ then for newly arriving transactions $rt_{\text{est}}$ 
is expected to be $rt_{\text{est}} = \o(mt) \times (q+1)$ because of the mostly sequential execution. 
Following this approach \Pe can balance commit rate and response time. 

Transaction termination triggers adaptation, however, it is important to note
that the adaptation is not executed as part of a transaction. This prevents the
situation where a failed adaptation would cause the transaction to abort, too.

\begin{myrule}Adaptation $O \rightarrow P$ with barrier 
\label{rule:adapt_O_P_with_barrier}
\item Let $cr$ be the commit rate, $\delta$ the hysteresis, $\gamma$ the target
commit rate, and $\beta$ the response time barrier. Adaptation is according to the following rules:

\begin{enumerate}
\item ($O \rightarrow P$): If $O$ is the current class for an item $x$ and $cr <
\gamma - \delta$ and $rt_{\text{est}} < \beta$ then $P$ is the new
classification of $x$.

\item ($P \rightarrow O$): If $P$ is the current class for an item $x$ and $cr$ is low ($cr < \gamma - \delta$) and $rt_{\text{est}} > \beta$ \\ 
 then $O$ is the new classification of $x$.

\item ($P \rightarrow O$): If $P$ is the current class for an item $x$ and $cr$ is high ($cr > \gamma + \delta$) \\
 then $O$ is the new classification of $x$.

\item Reclassification during a transaction: \\ a) If a $ta$ reads at the time when $O$ is the current class, but is about to write at a time when $P$ is the current class, $ta$ is aborted (non-avoidable crash) to maintain consistency. \\ b) If a $ta$ reads at a time when the data item is in $P$ and writes when it is in $O$ the success of the write is guaranteed because the data is exclusively locked since read-time.
\end{enumerate}
\end{myrule}

Rule~ \ref{rule:adapt_O_P_with_barrier}, 1) takes care that the commit rate is sufficiently high as long as the response time is low. If the response time exceeds the limit $\beta$ and $cr$ is (still) low then Rule~ \ref{rule:adapt_O_P_with_barrier}, 2) switches back to $O$. Rule~ \ref{rule:adapt_O_P_with_barrier}, 3) ensures that when the commit rate is high the default CC-mechanism of class $O$ is chosen. For all other situations the classification remains unchanged.

Rule \ref{rule:adapt_O_P_with_barrier}, 4) is the same as before. It ensures that a reclassification can take place during ongoing transactions. Reclassification is now triggered by two parameters, the commit rate $cr$ and the mean response time $mrt$. 

\section{Performance under Adaptation}
\label{sec:perform_adapt}

The performance study uses the implementation of \Pe described in Section
\ref{sec:prototypical_ref_impl}. Even if it is not a full database
implementation with all features (no backup and no recovery functionality) it is
sufficient for measuring the performance of \Pe under different situations. Since backup and recovery are normally inactive there is no impact on the concurrency mechanism. Therefore, the performance measurements would also be valid for a fully featured database system. Clearly, if backup or recovery are active, this would impair performance. This would also apply to our prototype.

The study analyzes different workload profiles indicated by a sequence of
workloads with a total life-span of one second each.
The workload is held constant for one second (called \emph{epoch}).
The arrival rate $\lambda$ for the workload ranges from $6.66$ tas/sec up to a
heavy overload of over $300$ tas/sec.
These values have been chosen, to show the behavior of the overloaded system
with frequent aborts and the behavior under moderate workload with a stable
commit rate.

During one epoch (1 sec) the commit rate is measured 10 times (sample rate $sr = 10$/sec). For simplicity, all transactions read and write only one data item, i.e., the worst case is simulated where an item in $O$ suddenly becomes a bottleneck. The time unit in all simulations is milliseconds if not stated otherwise.

To obtain a preliminary understanding the first experiments study short living
transactions with no disconnect time during three epochs. Afterwards long living
transactions with a random disconnect time $dt$ between 100 and 1000
milliseconds are analyzed over seven epochs.
A disconnect time $dt$ within these bounds simulates typical situations.

Finally, barrier $\beta$ is enabled for the next set of experiments. The set up of long living transactions and seven epochs is always the same except for the response time barrier $\beta$ which varies between $1000$ and $15000$ msec. We
study the effects on commit rate $cr$ and response time $rt$. Each experiment was executed three times.

\subsection{Short Living Transactions with Three Epochs}
\label{ssec:short_living_transactions_and_three_epochs}

Table \ref{tab:short_tas_3epochs} lists our test scenarios and summarizes the result. The right column of the table refers to the corresponding figures for a detailed analysis. The four tests use a different arrival rate $\lambda$ for each epoch (one second interval) as marked in the Epochs column. The first two test scenarios do not require a concurrency control adaptation to demonstrate the base performance without adaptation. In Tests \#3 and \#4 the workload is increased to trigger adaptation. 

\begin{table}[htbp]
  \centering
  \caption{Results for three epochs with different workload, $\gamma = 0.9$, $\delta = 5$\% and $dt=0$.}
    \begin{tabular}{rrrrrr}
    \toprule
    \multicolumn{6}{c}{Summary} \\
    \midrule
    Test\# & Epochs & \o$(cr)$   & $\sigma(cr)$   & \o$(rt)$ & Figure \\
    1&  9,14,19  & 1,00   & 0,00  & 3,6  & \ref{fig:adaptation_12-15} (a) \\
    2& 153,176,176 & 0,89 & 0,16  & 2    & \ref{fig:adaptation_12-15} (b) \\
    3& 10,19,178 & 0,96   & 0,06  & 3,7  & \ref{fig:adaptation_12-15} (c) \\
    4& 168,310,309 & 0,90 & 0,05  & 2824 & \ref{fig:adaptation_12-15} (d) \\
    \bottomrule
    \end{tabular}%
  \label{tab:short_tas_3epochs}%
\end{table}%

The average response time \o$(rt)$ is very high for test scenario \#4.
This is the result of an increasing overload, which quickly triggers adaptation
at the beginning of the second epoch (see Figure \ref{fig:adaptation_12-15}
(d)). This leads to a mostly sequential execution of the transactions, which
explains the very high \o$(rt)$ and the high \o$(cr)$ at the same time.
This increase of $cr$ is typical for scenarios after adaptation to $P$ has taken
place. It continues until the upper bound $\gamma + \delta$ is reached.
Then the adaptation switches back to class $O$.

As the tests indicate later, it would be better to add an additional criteria
for the re-adaptation from $P \rightarrow O$.
If the workload is still high (wait queue $> 1$) the data should remain in $P$
until the workload is low again before going back to $O$.
This measure could avoid multiple re-adaptations that produce an unstable system
behavior during a sudden transition of the workload from heavy overload to low
workload.

Figure \ref{fig:adaptation_12-15} shows the commit rate $cr$, lower and upper
bounds (set by $\gamma \pm \delta$), and the accumulated number of aborts and
commits of the four test scenarios.

\begin{figure*}[ht]
\centering%
\includegraphics[scale=0.7]{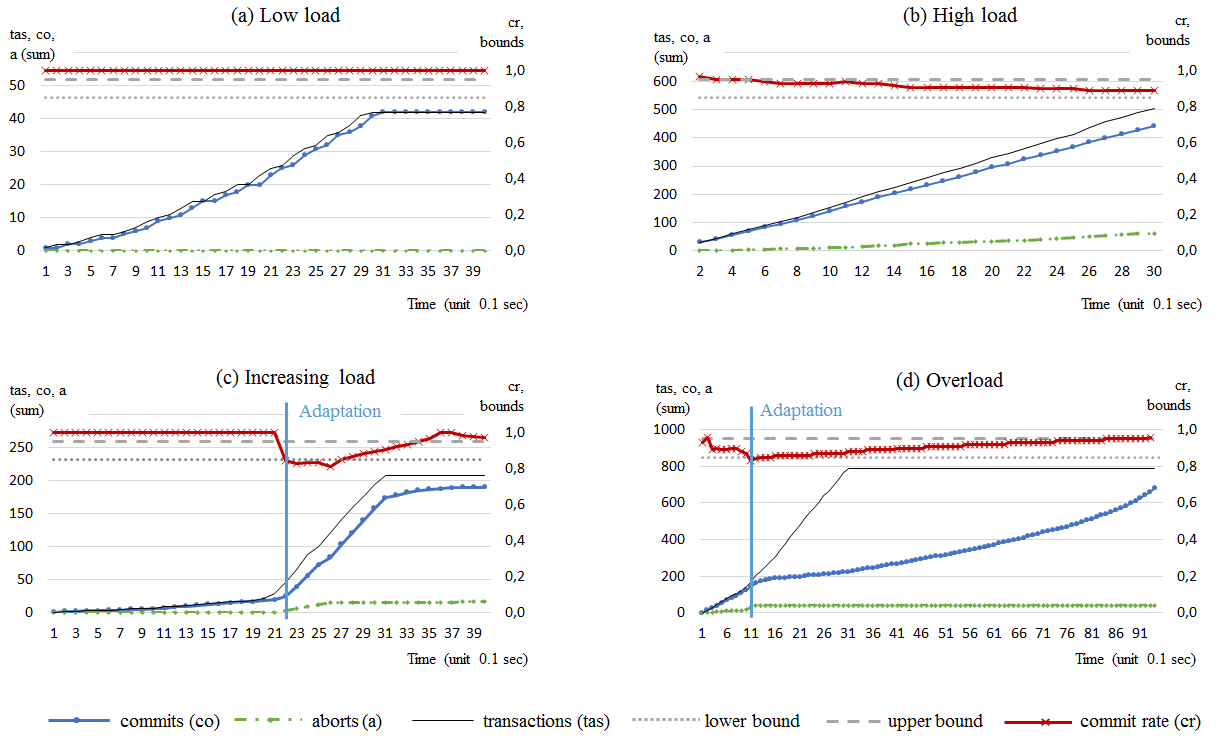}
\caption{Various short workloads to demonstrate Run-time Adaptation; (a) low $9-19$ tas/sec , (b) high $\approx 160$ tas/sec, (c) increasing load $10-180$ tas/sec, (d) increasing overload $170-310$ tas/sec, $\gamma=90\%$, $\delta=5\%$, $sr=10/sec$, and $dt=0$.}
\label{fig:adaptation_12-15}
\end{figure*}

Test \#1 has a low workload in all three epochs. The load starts with 9 tas/sec,
continues in epoch \#2 with 14 tas/sec and in the last epoch the workload rises
to 19 tas/sec. The transactions are executed as they arrive and no concurrent
interleaving transactions occur. As expected, no adaptation takes place. From
the corresponding Figure \ref{fig:adaptation_12-15} (a) it can be seen that the
commit rate is 1 and no aborts occur.
After $3.1$ sec (31 time units) all transactions have successfully terminated
and the number of commits remain constant.
Test \#1 is the only scenario without contention but surprisingly not the
shortest $rt$. The reason for this is that a commit is more expensive than an
abort for an optimistic CC. Compared to the other tests, Test \#1 has no aborts
and a commit rate of 100\%.

For Test \#2 the load is high ($\approx 160$ tas/sec) and nearly constant for three seconds. The load is heavy and contention is present as can be seen from the number of aborts and the decreasing commit rate. Figure \ref{fig:adaptation_12-15} (b) shows that the commit rate does not fall below the re-classification limit, hence no adaptation occurs. 
The data remains in class $O$ and the optimistic CC has low overhead which results in a short response-time of only $2$ msec. 

Part (c) of Figure \ref{fig:adaptation_12-15} (Test \#3) shows the results for
an increasing workload where finally in the third epoch the adaptation is
triggered. The workload starts with $10 - 19$ tas/sec for two seconds and
continues with $178$ tas/sec for the third epoch. The commit rate drops under
the minimum threshold ($\gamma - \delta$) at the blue vertical line ($2.2$ sec
after start).
The CC-mechanism immediately switches to locking and the number of aborts
decreases (the accumulated abort graph makes a sharp bend to a lower gradient).
During the third epoch the workload is slightly higher than the system can
immediately execute.
This can be seen from the slowly growing gap between the accumulated transaction
arrival (tas) and the accumulated committed transactions (co).
The average response time \o$(rt)$ stays low since during the first two seconds
the transactions were executed under $O$ with short $rt$.

It is interesting to compare Tests \#2 and \#3. Test \#2 has a constant high
workload, but not high enough to trigger the adaptation, hence, the data remains
in $O$.
This is the reason for the very short response time. Test \#3 has initially a
low workload, but in Epoch \#3 the workload just exceeds the threshold and
adaptation to $P$ applies.
This leads to a higher $rt$ even if the average workload is below the workload
of Test \#2.

Also, a start with low load (Tests \#1 and \#3) reduces the
response-time because all transactions of the first epoch are executed under optimistic CC with a short $rt$.

Test \#4 produces a heavy and increasing overload which triggers adaptation at the end of epoch 1. 
The gap between committed and arrived transactions grows until the arrival ends after 3 seconds. 
The adaptation to $P$ allows to increase the commit rate until after 10 sec all queued transactions have terminated. 
The system needs 7 sec to process the queued transactions after the arrival of transactions has stopped before it becomes resilient. 
This explains the high mean response time \o$(rt)$.

It can be noted that run-time adaptation under heavy workload achieves an average commit rate \o$(cr)$ of approximately 
 $90\%$, which was preset by $\gamma$. 
The price for improving $cr$ is clearly a longer response-time $rt$ which grows to $2.8$ seconds for continuous overload in test-case \#4. 

 The commit rate $cr$ is the basis for adaptation. When $cr$ drops below the lower bound $\gamma-\delta$ the adaptation is triggered and $cr$ increases again. The commit rate $cr$ increases until the upper bound $\gamma + \delta$ is reached which again triggers re-classification. 

Summarizing, for sudden increases and decreases of $cr$, adaptation ensures
a good response-time and a high commit rate if transactions are short lived
($dt=0$) and the system is not permanently overloaded. 
If contention constantly remains high, adaptation has severe effects
on the response-time.

\subsection{Long Living Transactions, Seven Epochs, and $\beta$ disabled}
\label{ssec:long_living_transactions_and_seven_epochs_beta_disabled}

Long living transactions are characterized by a certain time interval between
the read phase and the write phase where no data access occurs.
Some authors
\cite{Laux2009SQL} \cite{DBLP:conf/sigmod/AdyaGLM95} \cite{DBLP:conf/cloud/DingKDG15} \cite{DBLP:conf/vldb/HuangSRT91} \cite{DBLP:journals/tods/AgrawalCL87} \cite{Laux2010}
call this interval "think time" when a typical transaction reads and displays
data, then the user thinks about it, and finally modifies or adds some values.
We prefer to call this time "disconnect time", because Web based transactional
systems tend to logically disconnect from the database during this period.

For the tests a disconnect time $dt$ from $100$ - $1000$ msec was randomly chosen. Each test consisted of seven epochs with different workloads.
Workload W1 starts with $\lambda = 7 - 14$ tas/sec and rises the workload in epochs $3 - 7$ from 80 tas/sec continuously to 106 tas/sec. 
Workload W2 stresses the system with an increasing overload from $66 - 460$ tas/sec.

The detailed workload profiles are as follows:
\begin{itemize}
\item W1=(7,14,80,87,93,100,106) and
\item W2=(66,132,200,265,332,400,460).
\end{itemize}
Each number denotes the transactions arriving during the respective Epoch of one second each.
The tests were executed with two target commit rates $\gamma = 0.9$ and $0.7$.
Table \ref{tab:seven_epochs_no_beta} shows a summary of the results.

\begin{table}[htbp]
  \centering
  \caption[Results of seven epochs.]{Results of seven epochs with workload
  W1$=(7,14,80,87,93,100,106)$, W2$=(66,132,200,265,332,400,460)$, $\gamma =(0.9,0.7)$, random disconnect time $dt=100-1000$ ms, and barrier $\beta$ disabled.}
    \begin{tabular}{rrrrrrr}
    \toprule
    \multicolumn{1}{c}{Workload} & \multicolumn{1}{c}{$\gamma$} & \multicolumn{1}{c}{\o$(rt)$} & \multicolumn{1}{c}{\#Tas} & \multicolumn{1}{c}{\o$(cr_{\text{eff}})$} & Figure \\
    \midrule
    W1    & 90\%  & 4561  & 487   & 82\% & \ref{fig:sample_1_seven_epoches_w1}\\
    W2    & 90\%  & 24104 & 1845  & 82\% & \ref{fig:sample_5_seven_epoches_w2}\\
    W1    & 70\%  & 927   & 483   & 57\% \\
    W2    & 70\%  & 18957 & 1845  & 46\% & \ref{fig:sample_51_seven_epoches_w2_70}\\
    \bottomrule
    \end{tabular}%
  \label{tab:seven_epochs_no_beta}%
\end{table}%

Adaptation from $O \rightarrow P$ causes a systematic abort of pending transactions originating in $O$. To take these aborts into account the effective commit rate is defined as:
\begin{equation}
cr_{\text{eff}} := \frac{\mbox{\# committed tas}}{\mbox{\# terminated tas}}
\label{eq:commit_rate_eff}
\end{equation}
The effective commit rate $cr_{\text{eff}}$ measures -as the name suggests- the
performance of the system as shown to the user and the previously defined commit
rate $cr$ is used to trigger adaptation, because this indicator is more
sensitive to the workload.
The effective commit rate $cr_{\text{eff}}$ reached our tests 82\% for the first and $\approx 50$\% for the second value of $\gamma$.
Note that without adaptation all experiments would
have a commit rate between 1 and 3 percent only due to the long living nature of 
the transactions and the higher conflict potential. 
This is also the reason why the performance in this test scenario is lower than  in the previous subsection without disconnect time.

W1 has the shortest response-time due to the comparatively low workload.
In Epoch 3 with high workload ($80$ tas/sec) quickly lets $cr$ drop under the
lower boundary $\gamma-\delta = 0.85$ (see Figure
\ref{fig:sample_1_seven_epoches_w1}).
The adaptation to $P$ is triggered and in the following epochs $cr$ rises again
until the upper boundary is reached.
The data is reclassified in $O$ after 11 epochs and again the $cr$ drops, but
recovers faster as before, because the arrival of new transactions stopped after
7 epochs and after 12 epochs all pending transaction have terminated.

\begin{figure*}[ht]
\centering
\includegraphics[scale=0.6]{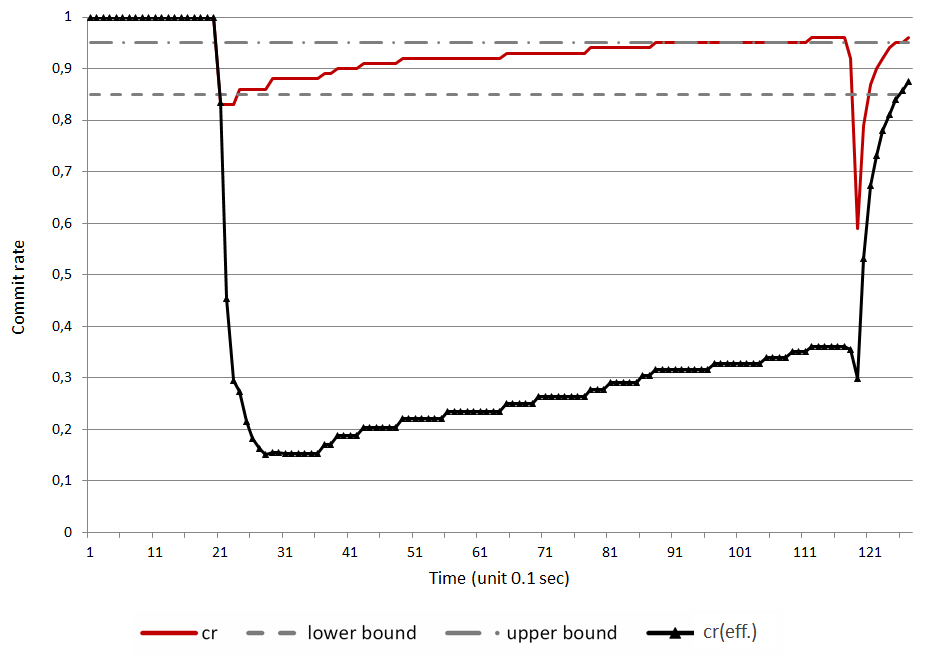}
\caption{Run-time adaptation for W1 $=(7,14,80,87,93,100,106)$,
$\gamma =0.9$, random disconnect time $dt=100-1000$ ms, and barrier $\beta$ disabled.}
\label{fig:sample_1_seven_epoches_w1}
\end{figure*}

The adaptation profile for workload W2 (permanent contention) shown in Figure
\ref{fig:sample_5_seven_epoches_w2} is similar to W1. Due to the heavy
workload starting in Epoch 1, the adaptation is already triggered at the end of
Epoch 1. The permanent overload leads to a significantly longer mean
response-time due to locking and queuing in $P$.

\begin{figure*}[ht]
\centering
\includegraphics[scale=0.6]{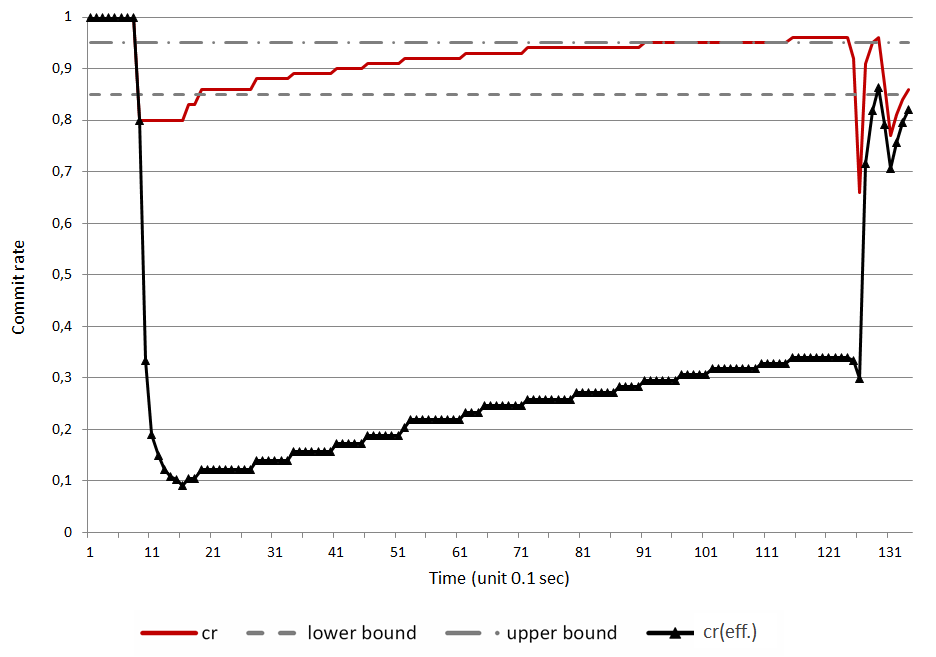}
\caption{Run-time adaptation for W2 $=(66,132,200,265,332,400,460)$, $\gamma =0.9$,
random disconnect time $dt =100-1000$ ms, and barrier $\beta$ disabled.}
\label{fig:sample_5_seven_epoches_w2}
\end{figure*}

For workload W2 (permanent contention, second row), the mean response-time is
significantly longer due to the queuing effect under $P$. Taking the same
workload with a target commit rate of $\gamma$ = 70\% the adaptation behavior
shows an instability (Figure \ref{fig:sample_51_seven_epoches_w2_70}).
After adaptation to $P$, the upper boundary for $cr$ is reached very quickly
during the third epoch (time = 27 units = 2.7 sec) and the data is reclassified
again in $O$ (Rule~ \ref{rule:adapt_O_P}, 2)) with the result that the commit rate
$cr$ drops to 40\%. After this decrease, the system recovers slowly and reaches
the upper boundary in epoch 14 again. At this point the arrival of transaction
has already stopped but the remaining (queued) transactions cause another jitter
for $cr$.

\begin{figure*}[ht]
\centering
\includegraphics[scale=0.6]{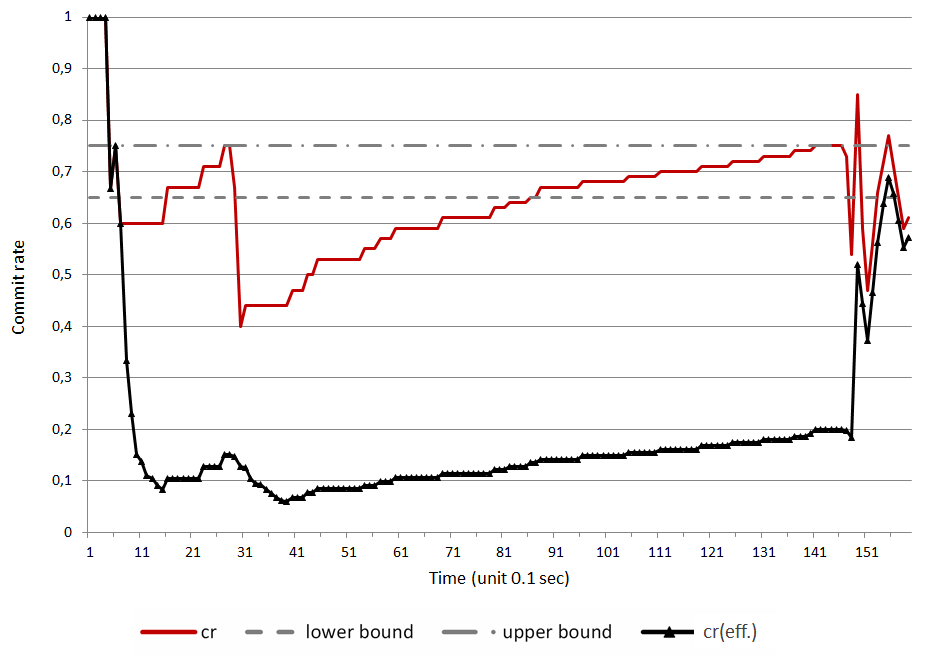}
\caption{Run-time adaptation for W2 $=(66,132,200,265,332,400,460)$, $\gamma =0.7$,
random disconnect time $dt =100-1000$ ms, and barrier $\beta$ disabled.}
\label{fig:sample_51_seven_epoches_w2_70}
\end{figure*}

The reason for this oscillating effect is that Rule~ \ref{rule:adapt_O_P}, 2) does
not look at the number of queued transactions it only takes criteria $cr >
\gamma + \delta$ to re-classify the data in $O$ again. But in this situation all
pending transactions except one will fail due to concurrency violation. This
lets the commit rate $cr$ drop as low as 40\%.

It takes now longer for the adaptation mechanism to reach the upper boundary
because many transactions have already aborted and accordingly more transaction
have to commit to rise $cr$. The upper boundary is reached after 14 sec when the
arrival of transaction has already stopped.

Summarizing, despite a sudden increase in contention, adaptation keeps the
commit rate stable even if transactions are long living. If contention remains
high, the response-time is getting longer since $P$ queues
transactions. 
With a low $\gamma$ the mechanism tends to become unstable and an oscillating behavior can be noticed. Having $\gamma$
close to 100\% is recommended since adaptation is triggered earlier. To prevent
an excessive increase in response-time, $\beta$ has to be enabled as
discussed in the next section.

\subsection{Long Living Transactions, Seven Epochs, and $\beta$ enabled}
\label{ssec:long_living_transactions_and_seven_epochs_beta_enabled}
The following experiments study the effects on the workloads of the previous subsection if barrier $\beta$ is enabled and $\gamma$ is high (=90\%) as recommended before. 
Table \ref{tab:seven_epochs_beta} summarizes the results and shows barrier
$\beta$, mean $cr_{\text{eff}}$, and the mean response-time \o$(rt)$ for  workloads W1 and W2. 
It further links to Figures \ref{fig:W1_seven_epoches_beta} and
\ref{fig:W2_seven_epoches_beta} showing sample graphs of
one run of an experiment at a time.

\begin{table}[hbp]
  \centering
  \caption[Results of seven epochs.]{Results of seven epochs with workload
  W1$=(7,14,80,87,93,100,106)$, W2$=(66,132,200,265,332,400,460)$, $\gamma=(0.9)$, random disconnect
  time $dt=100-1000$ ms, and barrier $\beta$ enabled.}
    \begin{tabular}{rrrrrr}
    \toprule
 Workload & $\beta$ & \o$(rt)$ & \o$(cr_{\text{eff}})$ & Figure\\
    \midrule
    W1    & 1000  & 187  & 17\% & Figure \ref{fig:W1_seven_epoches_beta} (a)\\
    W1    & 3000  & 343  & 18\% & Figure \ref{fig:W1_seven_epoches_beta} (b)\\
    W1    & 5000  & 355  & 29\% & --\\
    W1    & 8000  & 1960 & 36\% & Figure \ref{fig:W1_seven_epoches_beta} (c)\\
    W1    & 15000 & 3758 & 39\% & --\\
    W2    & 1000  & 136  & 3\%  & Figure \ref{fig:W2_seven_epoches_beta} (a)\\
    W2    & 3000  & 248  & 16\% & Figure \ref{fig:W2_seven_epoches_beta} (b)\\ 
    W2    & 5000  & 1219 & 18\% & --\\
    W2    & 8000  & 1172 & 25\% & Figure \ref{fig:W2_seven_epoches_beta} (c)\\
    W2    & 15000 & 2625 & 31\% & --\\
    \bottomrule
    \end{tabular}%
  \label{tab:seven_epochs_beta}%
\end{table}%  

As Table \ref{tab:seven_epochs_beta} shows, each workload was executed with
different values ($1000, 3000, 5000, 8000, 15000$) for $\beta$.
All experiments show that the mean response-time is bounded by $\beta$ and the
effect of a very long response-time of 19 or 24 seconds
 (see Table \ref{tab:seven_epochs_no_beta} of the previous section's
 experiments)
with workload  W2 no longer occurs.
The table also shows that the value of $\beta$ does not allow to infer the
actual mean response-time.
However, it shows that for an increasing $\beta$, the response-time and the
commit rate increase and $\beta$ correlates with these values.

Barrier $\beta$ does not directly match with the maximum response-time as given,
for example, by Service Level Agreements (SLA).
The response time depends on the workload and is directly influenced by the
transactions' arrival rate. The distribution of the response time depends
additionally on the concurrency model.
For a queuing system like the concurrency model of class $P$ a Poisson arrival
process is assumed. The response time $rt$ is calculated as wait time $wt$ in
the queue plus transaction processing $pt$ time.
Even in the simplest queuing system, the P/P/1, with Poisson arrival and one
service process, only statements about the mean response time \o$(rt)$ can be
made. To estimate the expected response time $rt_{\text{est}}$, the arrival and
service rate is necessary. But in the present case both rates are heavily
changing.
If the arrival rate would only change due to statistical variation no adaptation
would be necessary. But if a systematic change happens, e.g., because the data
access type changes, the original class assignment is not any more suitable.
 Adaptation changes the service time and hence the service rate as well. The
 service time $st$ in the case of $P$ is the time between read and write.
 The only indicators for the estimated response time are the wait queue length
 $|Q_w|$ and the past average  \o$(st)$.
This leads to Formula (\ref{eq:expect_rt}):
\begin{equation}
rt_{\text{est}} := \o(st) \times (|Q_w| + 1)
\label{eq:expect_rt}
\end{equation}

The calculation takes into account the transactions that are already queued for
execution and the average time to process a transaction. The processing time
includes a possible waiting time due to locking.

The SLA defines a limit for the response time $rt$ and in the case of an SLA
violation, a penalty has to be paid.
There is a trade off between loosing transactions or having excessive response
time.
Assuming an average price of $r$ for each lost transaction and a penalty of $p$
for every transaction exceeding the response time limit $\beta$ the trade-off is
given at the intersection of two cost functions that depend on the commit rate
$cr$ and the number of transactions $tas$:
\begin{eqnarray}
ca := r \times (1-cr) \times tas  \\
\label{eq:cost_abort}
cp := p \times tas_{rt > \beta}(cr) \times tas
\label{eq:cost_sla}
\end{eqnarray}

If the functions are normalized with the number of transactions $tas$ then Figure \ref{fig:trade_off} shows the principal graph for this trade off. 
The break even point for this normalized example is given at commit rate $cr =0.72$. In practice, the database system will measure the actual and number of aborts and the application should monitor these values and calculate the break even based on the costs for SLA violation and failed transactions.
 
\begin{figure}[h]
\centering
\includegraphics[scale=0.55]{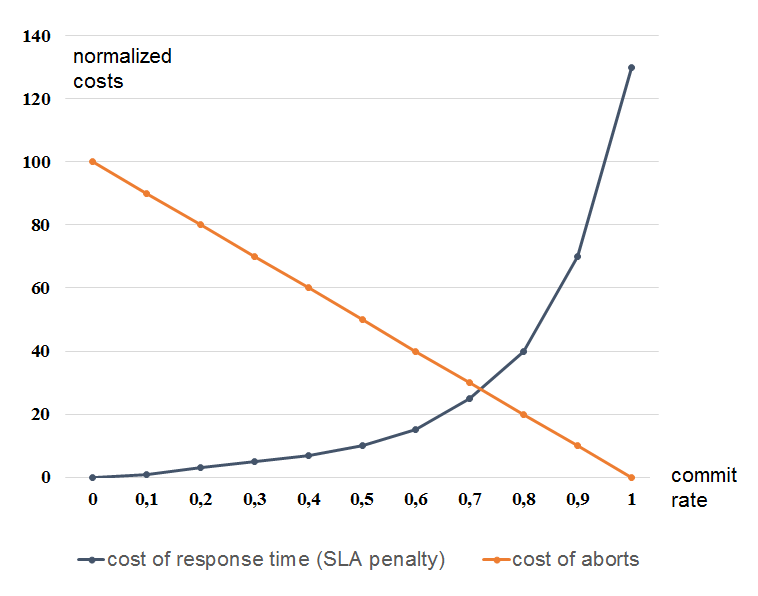}
\caption{Example trade off between aborts and response time in terms of costs.}
\label{fig:trade_off}
\end{figure}

In the case of a fast changing workload it is difficult to estimate the workload
profile. If the calculation is based on the past workload, the system may not
react fast enough to sudden changes of the arrival rate.

The situation is more promising if a workload profile is known in advance.
This is often the case if employees have clear routines during their workday.
Assume, for example, the following tasks: order processing in the morning, stock
administration after lunch, and master data management from 5 pm to 6 pm. In
this scenario data access to product data in the morning and afternoon will be
classified $O$ while the product data will be re-classified to $P$ from 5 - 6
pm.

\begin{figure*}
\centering
\includegraphics[scale=1.0]{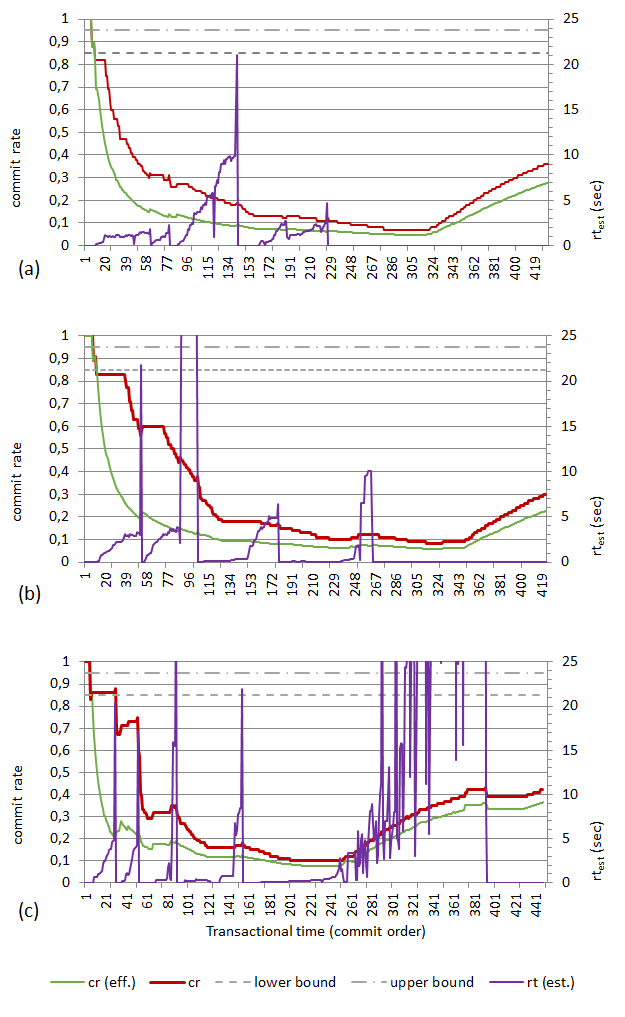}
\caption{Run-time adaptation profile for target commit rate $\gamma=0.9$: (a) workload W1 with $\beta=1000$, (b) Workload W1 with $\beta=3000$, and (c) Workload W1 with $\beta=8000$.}
\label{fig:W1_seven_epoches_beta}
\end{figure*}

\begin{figure*}
\centering
\includegraphics[scale=1.0]{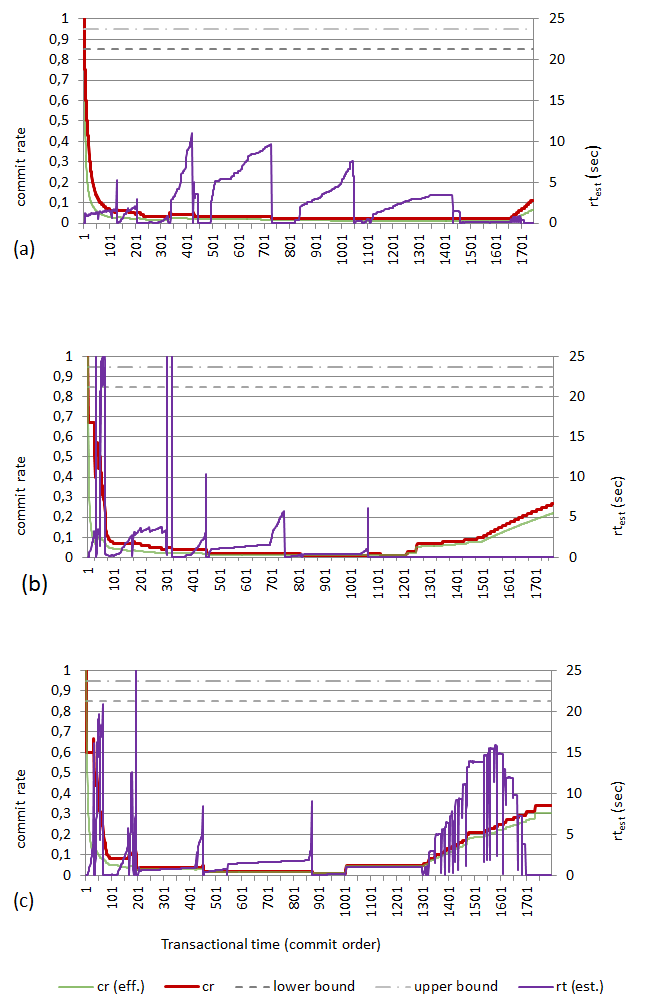}
\caption{Run-time adaptation profile for target commit rate $\gamma=0.9$: (a) workload W2 with $\beta=1000$, (b) Workload W2 with $\beta=3000$, and (c) Workload W2 with $\beta=8000$.}
\label{fig:W2_seven_epoches_beta}
\end{figure*}

Figures \ref{fig:W1_seven_epoches_beta} and \ref{fig:W2_seven_epoches_beta}
illustrate the run-time adaptation profile if $\beta$ is set. The time of the
estimated response time $rt_{\text{est}}$ is shown on the right vertical axis.
The left ordinate shows the commit rate and the target boundaries.
The horizontal axis shows the transactional time which is given by a sequence of
time ordered events. The time interval from one event to the next is not
constant and hence the time scale is not linear.

In Figure \ref{fig:W1_seven_epoches_beta} (a) the commit rate $cr$ is $1$ during
the first two seconds when the workload is low.
When the overload begins after two seconds the commit rate $cr$ drops quickly
below the lower bound $\gamma-\delta = 0.85$ and adaptation to $P$ takes place.
The effective commit rate \crb (green line in Figures
\ref{fig:W1_seven_epoches_beta} and \ref{fig:W2_seven_epoches_beta}) always
stays below $cr$ because $cr$ does not count aborts due to the adaptation $O
\rightarrow P$, but \crb does.
After adaptation to $P$ the system stabilizes the commit rate $cr$ as shown by
the red graph.
This appears in all test runs and can be seen more clearly when we have a higher
response time limit $\beta$ as in part (b) and (c) of Figure
\ref{fig:W1_seven_epoches_beta}.

In the case of $\beta=8000$ (part (c)) the commit rate increases until the
estimated response time exceeds the preset limit $\beta$. If $rt_{\text{est}} >
\beta$ the re-adaptation to $O$ is triggered by Rule
\ref{rule:adapt_O_P_with_barrier}, 2) because $cr$ is still below the lower
bound.
The result is that pending transactions abort and in the following the commit
rate decreases.
Run time estimation works for $P$ only because a wait queue $Q_w$ is needed for
Formula (\ref{eq:expect_rt}).
If there is no wait queue then the number for \rtb is set to $0$.
Hence, as soon as \rtb exceeds $\beta$ the systems switches to $O$ and the \rtb
drops to 0, which explains the saw tooth figure of \rte.
  
A low limit for the response time as in Figure \ref{fig:W1_seven_epoches_beta}
(a) causes a low $cr$ and many transactions run in $O$, which can only be
observed indirectly by the low \rte.
If $\beta$ increases (Figure \ref{fig:W1_seven_epoches_beta} (b) and (c)), the
number of aborts reduces because the system remains longer in $P$, which causes
more waiting transactions which in turn cause more and higher peaks in the \rte.

For workload W2, Figures \ref{fig:W2_seven_epoches_beta} (a) - (c) illustrate the workload profiles for
$\beta=1000$, $\beta=3000$, and $\beta=8000$. 
The graph of $rt_{\text{est}}$  for $\beta=1000$ shows
regularly appearing peaks of longer duration caused by the permanent
contention. This effect nearly disappears for larger values of $\beta$ $(\geq 8000)$ because after adaptation to $P$ 
much more transactions are allowed to queue up and commit later. 
This is indicated by higher and shorter \rtb peaks, which move to the beginning of the test run.  
As a result the \crb is slightly higher if $\beta$ is high.

When the workload ends after 7 seconds and the \rtb drops below $\beta$ Rule~
\ref{rule:adapt_O_P_with_barrier}, 1) applies and the concurrency class switches
to $P$ which lets $cr$ and \crb rises until all transactions have terminated.

Part (c) of Figures \ref{fig:W1_seven_epoches_beta} and
\ref{fig:W2_seven_epoches_beta} show an effect of instability. This happens
after the arrival of transactions has ended and before all transactions have
terminated. The system has switched to $P$ because \rtb was below the limit
$\beta$ and now the high number of remaining transactions in the queue leads to
\rte $> \beta = 8000$ and the re-adaptation to $O$ lets \rtb drop below $\beta$,
which again triggers Rule~ \ref{rule:adapt_O_P_with_barrier}, 1) and forces the
data to class $P$. The oscillation between $P$ and $O$ continues until most
transactions have terminated and the queue is short enough to keep \rtb below
the limit $\beta$.

In part (a) and (b) this effect shows up in a moderate form during the workload but not after its termination because the lower $\beta$ does not allow many transactions to be queued and delayed for a longer time.   

Summarizing, the usage of $\beta$ keeps the mean response-time bounded, but
compared to having $\beta=\infty$, a higher abort rate is the price that has to be paid.
The exact determination of $\beta$ demands a continuous adjustment and has to be
carried out by applications. In particular in the case of a mixed workload a
greater $\beta$ causes short peaks in the $rt_{\text{est}}$ since more transactions are allowed to commit in $P$. 
A lower $\beta$ causes longer peaks since many transactions wait and their abort
is not yet known. They continue in $O$ and abort at write-time at the earliest.

Generally, it is important to know that \Pe classifies hot spot items (HSIs) in
classes $R$ and $E$, if possible. This is the better choice if the semantic of the data allows this classification.
Adaptation is only provided to handle a sudden,
but impermanent increase of the contention for items classified in default class $O$. 
Permanent contention is likely to cause any system to become overloaded. \Pe is at least
able to protect itself by trading off response-time and commit rate.

\section{Related Work}
\label{sec:related_work}
This paper extends the findings of \cite{LessnerDBKDA2015} and is based on the Ph.D. thesis \cite{Lessner2014} of the main author, which introduces \PeFS. A vast amount of work \cite{JimGray.1993} \cite{GerhardWeikum.2002} has been
carried out in the field of transaction management and CC, but so far no attempt
was undertaken to use a combination of CC mechanisms according to the data usage
(semantics). Most authors use the semantics of a transaction to divide it into
sub-transactions, thus achieving a finer granularity that hopefully exhibit less
conflicts. Some authors \cite{Garcia-Molina1983} use the semantics of the data
to build a compatibility set while others try to reduce conflicts using
multiversions \cite{Phatak2004} \cite{Graham1994}. The reconciliation mechanism was
introduced in \cite{Laux2009} and is an optimistic variant of ``The escrow transactional method'' \cite{O'Neil1986}. Escrow relies on guaranties given to the
transaction before the commit time, which is only possible for a certain
class of transactions, e.g. transactions with commutative operations. Optimistic concurrency control was introduced by \cite{Kung.1981}, which did not
gain much consideration in practice until SI, introduced by \cite{Berenson1995},
has been implemented in an optimistic way.
SI in general gained much attention through \cite{Fekete2005a} \cite{Cahill2009}, and also in practice \cite{Ports2012}. Its strength lies in applications that have to deal with many concurrent queries but has only a moderate rate of updating transactions. 
\PeFS, however, is designed for high performance updating transactions processing, especially with data hot spots.

\section{Conclusion and Outlook}
\label{sec:outlook}
The paper presented a multimodel concurrency control mechanism that breaks with the \emph{one concurrency mechanism fits all needs}. The concurrency mechanism is chosen according to the access semantic of the data. Four concurrency control classes are defined and rules guide the developer with the manual classification. When the access semantic is unknown the default class $O$ with an optimistic snapshot isolation mechanism is chosen. For those data the model is extended to dynamically change the class assignment if the performance suggests a pessimistic mechanism $P$. The simulations with the prototype demonstrated that the mechanism is working and tests with the TPC-C++ benchmark  resulted in a 3 to 4 times superior performance. The adaptation mechanism provides a response time guaranty to comply with Service Level Agreements for
the price of a lower commit rate.   

The tests revealed an instability in the form of an oscillating adaptation. This occurs only under an abrupt change of the workload from overload to inactive system.
However, a refinement of the adaptation rule could possibly avoid the oscillation when the re-classification from $P \rightarrow O$ is executed. This could be achieved if the re-classification is only triggered when the wait queue is small or empty.

A dynamic algorithm for an automatic classification of data would be
desirable and would relief the developer from manual classification. The same mechanism could then be used to dynamically adapt the data according to a changed usage profile. 

Also, comprehensive performance tests that consider replication, online backup and a study of run-time adaptation under real-life conditions is still missing.
\newpage

\end{document}